\documentclass{lmcs} 
\pdfoutput=1

\usepackage{lastpage}
\lmcsdoi{17}{4}{19}
\lmcsheading{}{\pageref{LastPage}}{}{}%
{Jun.~23,~2020}{Dec.~15,~2021}{}

\usepackage[utf8]{inputenc}


\theoremstyle{plain} 

\usepackage{amsmath}
\usepackage{stmaryrd}
\usepackage{amsfonts}
\usepackage{amssymb}
\usepackage[all]{xy}

\usepackage{todonotes}

\newcommand{\Cat}[1]{\mathcal{#1}}
\newcommand{\Set}{\mathsf{Set}}
\newcommand{\SL}{\mathsf{SL}}

\newcommand{\CLat}{\mathsf{CLat}_{\wedge}}
\newcommand{\Pos}{\mathsf{Pos}}

\newcommand{\Pow}{\mathcal{P}}

\newcommand{\Powf}{\Pow_\omega}
\newcommand{\Id}{\mathsf{Id}}
\newcommand{\Eq}{\Delta}
\newcommand{\Quot}{\mathsf{Quot}}

\newcommand{\id}{\mathsf{id}}
\newcommand{\op}{\mathsf{op}}

\newcommand{\theory}{\mathit{th}}
\newcommand{\Sem}[1]{\llbracket #1 \rrbracket}
\newcommand{\diam}[1]{\left\langle #1 \right\rangle}

\newcommand{\dlog}{d_{\mathit{log}}}
\newcommand{\dsdw}{d_{\mathit{sdw}}}

\newcommand{\lift}[1]{\overline{#1}}

\newcommand{\Alg}{\mathrm{Alg}}
\newcommand{\Rel}{\mathrm{Rel}}
\newcommand{\RelV}{\mathrm{Rel}_{\V}}
\newcommand{\Pred}{\mathrm{Pred}}
\newcommand{\V}{\mathcal{V}}

\newcommand{\Hom}{\mathrm{Hom}}
\newcommand{\strength}{\mathrm{st}}
\newcommand{\TV}{\Omega}
\newcommand{\Dis}{D}

\newcommand{\logic}{\mathcal{L}}
\newcommand{\pop}{\mathcal{O}}
\newcommand{\coloneqq}{\mathrel{:=}}

\newcommand{\llift}[1]{\widehat{#1}}
\newcommand{\Gr}{\mathrm{Gr}}

\theoremstyle{plain}
\newtheorem{lemma}[thm]{Lemma}
\newtheorem{proposition}[thm]{Proposition}
\newtheorem{theorem}[thm]{Theorem}
\newtheorem{corollary}[thm]{Corollary}
\theoremstyle{definition}
\newtheorem{definition}[thm]{Definition}
\newtheorem{example}[thm]{Example}
\newtheorem{remark}[thm]{Remark}
\newtheorem{assumption}[thm]{Assumption}

\begin{document}

\title{Expressive Logics for Coinductive Predicates}
\titlecomment{{\lsuper*}
This is a revised and extended version of a paper which appeared in the proceedings
of CSL 2020~\cite{KupkeR20}.
Part of this research was carried out during the second author's stay at University College London,
funded by the European Union's Horizon 2020 research and innovation programme under the Marie Sk\l{}odowska-Curie Grant Agreement No.\ 795119.
}

\author[C.~Kupke]{Clemens Kupke}[a]	
\address{Department of Computer \& Information Sciences, 
Strathclyde University, United Kingdom}	
\email{clemens.kupke@strath.ac.uk}  

\author[J.~Rot]{Jurriaan Rot}[b]	
\address{Institute for Computing and
    Information Sciences (iCIS), Radboud University, The Netherlands}	
\email{jrot@cs.ru.nl}  





\begin{abstract}
The classical Hennessy-Milner theorem says that two states of an image-finite transition system
are bisimilar if and only if they satisfy the same formulas in a certain modal logic. 
In this paper we study this type of result in a general context,
moving from transition systems to coalgebras and from bisimilarity to coinductive predicates.
We formulate when a logic fully characterises a coinductive predicate on coalgebras, by
providing suitable notions of adequacy and expressiveness,
and give sufficient conditions on the semantics. 
The approach is illustrated with logics
characterising similarity, divergence and a behavioural metric on automata.  
\end{abstract}

\maketitle

\section{Introduction}
\label{sec:introduction}

The deep connection between bisimilarity and modal logic manifests itself in the \emph{Hennessy-Milner theorem}:
two states of an image-finite labelled transition system (LTS)
are behaviourally equivalent iff they satisfy the same formulas in a certain modal logic~\cite{HennessyM85}. 
From left to right, this equivalence is sometimes referred to as \emph{adequacy} of the logic w.r.t.\ bisimilarity, 
and from right to left as \emph{expressiveness}. By stating both adequacy and expressiveness, the Hennessy-Milner
theorem thus gives a logical characterisation of behavioural equivalence.  
 
There are numerous variants and generalisations of this kind of result. 
For instance, a state $x$ of an LTS is \emph{simulated} by a state $y$ if every formula 
satisfied by $x$ is also satisfied by $y$, where the logic 
only has conjunction and diamond modalities; see~\cite{Glabbeek90}
for this and many other related results. Another class of examples is logical characterisations
of quantitative notions of equivalence, such as probabilistic bisimilarity and behavioural 
distances (e.g.,~\cite{LarsenS91,DesharnaisGJP99,DesharnaisEP02,BreugelW05,JacobsS09,KonigM18,WildSP018,ClercFKP19}). 
In many such cases, including bisimilarity, the comparison between states is \emph{coinductive}, and the problem is thus
to characterise a coinductively defined relation (or distance) with a suitable modal logic. 

Both coinduction and modal logic can be naturally and generally studied within the theory of \emph{coalgebra},
which provides an abstract, uniform study of state-based systems~\cite{Rutten00,jacobs-coalg}. 
Indeed, in the area of \emph{coalgebraic modal logic}~\cite{KupkeP11}
there is a rich literature on deriving expressive logics for behavioural equivalence between 
state-based systems, thus going well beyond labelled transition systems~\cite{Pattinson04,schr08:expr,Klin07}.  
However, such results focus almost exclusively on behavioural equivalence or bisimilarity---a 
coalgebraic theory of logics for characterising coinductive predicates other than bisimilarity is
still missing. 
The aim of this paper is to accommodate the study of logical characterisation of coinductive
predicates in a general manner, and provide tools to prove adequacy and expressiveness. 

Our approach is based on universal coalgebra, to achieve results that apply generally to state-based 
systems. Central to the approach are the following two ingredients. 
\begin{enumerate}
	\item \emph{Coinductive predicates in a fibration.}
	To characterise coinductive predicates, we make use of fibrations---this approach 
	originates from the seminal work of Hermida and Jacobs~\cite{HJ98}. The fibration 
	is used to speak about predicates and relations on states.
	In this context, liftings of the type functor of coalgebras uniformly determine coinductive predicates and
	relations on such coalgebras. 
	An important feature of this approach, advocated in~\cite{HasuoKC18}, 
	is that it covers not only bisimilarity, but also other coinductive predicates including, e.g., 
	similarity of labelled transition systems and other coalgebras~\cite{HughesJ04}, 
	behavioural metrics~\cite{BaldanBKK18,Bonchi0P18,SprungerKDH18}, 
	unary predicates such as divergence~\cite{BonchiPPR17,HasuoKC18}, 
	and many more. 
	
	\item \emph{Coalgebraic modal logic via dual adjunctions.}
	We use an abstract formulation of coalgebraic logic, which
	originated in~\cite{PavlovicMW06,Klin07}, building on a tradition of logics via duality (e.g.,~\cite{KupkeKP04,BonsangueK05}). 
	This framework is formulated in terms of a contravariant adjunction, which captures
	the basic connection between states and theories, and a distributive law, which
	captures the one-step semantics of the logic. 
	It covers classical modal logics of course,
	but also easily accommodates multi-valued logics, and, e.g., logics without
	propositional connectives, where formulas can be thought of as basic tests on state-based systems. 
	This makes the framework suitable for an abstract formulation of Hennessy-Milner type theorems, where
	formulas play the role of tests on state-based systems. 
\end{enumerate}
To formulate adequacy and expressiveness with respect to general coinductive predicates, we need to know
how to compare collections of formulas. For instance, if the coinductive predicate is similarity of LTSs, 
then the associated logical theories of one state should be \emph{included} in the other, not necessarily equal. 
This amounts to stipulating a \emph{relation} on truth values, that extends to a relation between
theories. In the quantitative case, we need 
a \emph{logical distance} between collections of formulas; this typically arises
from a distance between truth values (which, in this case, will typically be an interval in the real numbers). 
The fibrational setting provides a convenient means for defining such an object for comparing theories.

With this in hand, we arrive at the main contributions of this paper: the formulation of adequacy and expressiveness
of a coalgebraic modal logic with respect to a coinductive predicate in a fibration,
and sufficient conditions on the semantics of the logic that guarantee adequacy and expressiveness. 
We exemplify the approach through a range of examples, including logical characterisations of a simple behavioural
distance on deterministic automata, similarity of labelled transition systems, and a logical characterisation
of a unary predicate: divergence, the set of states of an LTS which have an infinite path of outgoing $\tau$-steps. 
The latter is characterised, on image-finite LTSs, by a quantitative logic with only diamond formulas, i.e.,
the set of formulas is simply the set of words.

\subsection*{Related work}

As mentioned above, there are numerous specific results on Hennessy-Milner theorems, 
which---e.g., in the probabilistic setting as in~\cite{ClercFKP19}---can be highly non-trivial. A comprehensive
historical treatment is beyond the scope of this paper,
which is, instead, broad: it aims at studying these kinds of results in a general,
coalgebraic setting. 

The case of capturing bisimilarity
and behavioural equivalence of coalgebras by modal logics has been very well studied,
see~\cite{KupkeP11} for an overview.
Expressiveness w.r.t.\ similarity has been studied in~\cite{KapulkinKV12},
which is close in spirit to our approach, but focuses on the poset case. On a detailed level,
the logic for similarity is based on distributive lattices, hence
it uses disjunction; this differs from our example, which only uses conjunction and diamond modalities. 
Another study of expressiveness of logics w.r.t.\ various forms of similarity is in~\cite{soton354112}.
Expressiveness of multi-valued coalgebraic logics w.r.t.\ behavioural equivalence is studied in~\cite{BilkovaD16}.
In~\cite{BakhtiariH17}, notions of equivalence are extracted from a logic through a variant 
of $\Lambda$-bisimulation~\cite{GorinS13}. 
To the best of our knowledge, the current work is the first in the area that connects
general coinductive predicates in a fibration to coalgebraic logics. 

In the recent~\cite{graded}, the authors prove Hennessy-Milner type theorems for coalgebras including, but going
significantly beyond bisimilarity.
The logics are related to a semantics obtained from graded monads. 
The scope differs substantially from the current paper: the graded monad approach
is inductive and focuses on semantic equivalence of different types,
whereas our framework aims at characterising coinductive predicates, and
in our approach it is essential to be able to relate theories in different ways than equivalence 
(to cover, e.g., similarity, divergence or logical distance).
On the one hand, it appears
that none of our examples can be covered immediately in \emph{loc.\ cit.}; on the other hand, 
trace equivalence of various kinds can be covered in~\cite{graded} but not (directly) in the current paper. 
 
In~\cite{WildSP018} a characterisation theorem is shown for fuzzy modal logic, and 
in~\cite{KonigM18} for a wide class of behavioural metrics. 
These papers are not aimed at other kinds of
coinductive predicates, and they do not cover the examples in Section~\ref{sec:examples} (including
the behavioural metric for deterministic automata, as we use a much simpler logic than in~\cite{KonigM18}). 
Conversely, the question whether the logical characterisation results of~\cite{KonigM18}
can be covered in the current framework is left open. 
These papers also treat game-based characterisations of bisimilarity, which
are studied in a general setting in the recent~\cite{KKHKH19}. That paper, however,
does not yet feature modal logic explicitly; in fact, the connection is posed
there as future work. 

An earlier version of this paper appeared in the proceedings of CSL 2020~\cite{KupkeR20}.
The current paper extends this with a more detailed treatment of the basic setup of logics in a 
contravariant adjunction in Section~\ref{sec:adjunctions}, a treatment of expressiveness of logics w.r.t.\ behavioural
equivalence rather than bisimilarity making use of so-called lax liftings, in Section~\ref{sec:canonical}, and a 
new section on finite-depth expressiveness via initial and final sequences (Section~\ref{sec:kleene}).

\subsection*{Outline}
The paper starts in Section~\ref{sec:prelims} with preliminaries on coalgebra, fibrations and coinductive predicates, and coalgebraic
modal logic. Section~\ref{sec:framework} contains the abstract framework for expressiveness and adequacy, together with 
sufficient conditions for establishing these. Section~\ref{sec:examples} contains three detailed examples of this setup. A different
route to expressiveness, via initial and final sequences, is explored in Section~\ref{sec:kleene}. The paper
concludes in Section~\ref{sec:conclusions} with directions for future work.

\subsection*{Acknowledgments} We would like to thank Fredrik Nordvall Forsberg, Ichiro Hasuo, Bart Jacobs, Shin-ya Katsumata and
Yuichi Komorida for helpful discussions, comments and suggestions.

\section{Preliminaries}
\label{sec:prelims}

The category of sets and functions is denoted by $\Set$. The powerset functor
is denoted by $\Pow \colon \Set \rightarrow \Set$, and the finite powerset
functor by $\Powf$. The diagonal relation on a set $X$ 
is denoted by $\Delta_X = \{(x,x) \mid x \in X\}$.

Let $\Cat{C}$ be a category, and $B \colon \Cat{C} \rightarrow \Cat{C}$ a functor. 
A \emph{$(B)$-coalgebra} is a pair $(X,\gamma)$ where $X$ is an object in $\Cat{C}$ and
$\gamma \colon X \rightarrow BX$ a morphism. A homomorphism from a coalgebra $(X,\gamma)$
to a coalgebra $(Y,\theta)$ is a morphism $h \colon X \rightarrow Y$ such that $\theta \circ h = Bh \circ \gamma$. 
An \emph{algebra} for a functor $L \colon \Cat{D} \rightarrow \Cat{D}$ on a category $\Cat{D}$ is a pair $(A,\alpha)$ of
an object $A$ in $\Cat{D}$ and an arrow $\alpha \colon LA \rightarrow A$. 

\begin{example}
A \emph{labelled transition system (LTS)} over a set of labels $A$ is a coalgebra 
$(X,\gamma)$ for the functor $B \colon \Set \rightarrow \Set$, $BX = (\Pow X)^A$. 
For states $x,x' \in X$ and a label $a \in A$, we sometimes write $x \xrightarrow{a} x'$ for $x' \in \gamma(x)(a)$. 
Image-finite labelled transition systems are coalgebras for the functor $BX = (\Powf X)^A$.  
A \emph{deterministic automaton} over an alphabet $A$ is 
a coalgebra for the functor $B \colon \Set \rightarrow \Set$, $BX = 2 \times X^A$. 
For many other examples of state-based systems modelled as coalgebras, see, e.g.,~\cite{jacobs-coalg,Rutten00}.
\end{example}

\subsection{Coinductive Predicates in a Fibration}
\label{sec:coind-pred}

We recall the general approach to coinductive predicates in a fibration,
starting by briefly presenting how bisimilarity of $\Set$ coalgebras arises 
in this setting (see~\cite{HasuoKC18,HJ98,jacobs-coalg} for details). 
Let $\Rel$ be the category where an object is a pair $(X,R)$ consisting of a set $X$ and a relation $R \subseteq X \times X$ on it,
and a morphism from $(X,R)$ to $(Y,S)$ is a map $f \colon X \rightarrow Y$ such that $x \mathrel{R} y$ implies $f(x) \mathrel{R} f(y)$, for all $x,y \in X$. 
Below, we sometimes refer to an object $(X,R)$ only by the relation $R \subseteq X \times X$. 
Any set functor $B \colon \Set \rightarrow \Set$ gives rise to a functor $\Rel(B) \colon \Rel \rightarrow \Rel$,
defined by \emph{relation lifting}: 
\begin{equation}\label{eq:rel-lift}
\Rel(B)(R \subseteq X \times X) = \{((B\pi_1)(z), (B\pi_2)(z)) \in BX \times BX \mid z \in BR\} \,.
\end{equation}
Given a $B$-coalgebra $(X,\gamma)$, a \emph{bisimulation} is a relation $R \subseteq X \times X$ such that
$R \subseteq (\gamma \times \gamma)^{-1} (\Rel(B)(R))$, i.e., if $x \mathrel{R} y$
then $\gamma(x) \mathrel{\Rel(B)(R)} \gamma(y)$. \emph{Bisimilarity}
is the greatest such relation, and equivalently, the greatest
fixed point of the monotone map $R \mapsto {(\gamma \times \gamma)^{-1} (\Rel(B)(R))}$
on the complete lattice of relations on $X$, ordered by inclusion.

The functor $\Rel(B)$ is a \emph{lifting} of $B$: it maps a relation on $X$ to a relation on $BX$. 
A first step towards generalisation beyond bisimilarity is obtained by replacing $\Rel(B)$
by an arbitrary lifting $\lift{B} \colon \Rel \rightarrow \Rel$ of $B$. 
For instance, for $BX = (\Powf X)^A$ one may take
\begin{equation}\label{eq:bbar-simulation}
\overline{B}(R) = \{(t_1,t_2) \mid \forall a \in A. \, \forall x \in t_1(a). \, \exists y \in t_2(a). (x,y) \in R \}\,.
\end{equation} 
Then, for an LTS $\gamma \colon X \rightarrow (\Powf X)^A$, 
the greatest fixed point of the monotone map $R \mapsto (\gamma \times \gamma)^{-1} \circ \overline{B}(R)$
is \emph{similarity}. 
In the same way, by varying the lifting $\lift{B}$, one can define many different coinductive
relations on $\Set$ coalgebras. 

Yet a further generalisation is obtained by replacing $\Set$ by a general category $\Cat{C}$,
and $\Rel$ by a category of `predicates' on $\Cat{C}$. A suitable categorical
infrastructure for such predicates on $\Cat{C}$ is given by the notion of \emph{fibration}. 
This allows us, for instance, to move beyond (Boolean, binary)
relations to quantitative relations (e.g., behavioural metrics) or unary predicates. Such
examples follow in Section~\ref{sec:examples}; also see, e.g.,~\cite{HasuoKC18,BonchiPPR17}. 

To define fibrations, it will be useful to fix some associated terminology first. 
Let $p \colon \Cat{E} \rightarrow \Cat{C}$ be a functor. If $p(R) = X$,
then we say $R$ is \emph{above} $X$, and similarly for morphisms. 
The collection of all objects $R$ above a given object $X$ and arrows above the
identity $\id_X$ form a category, called the \emph{fibre above $X$} and
denoted by $\Cat{E}_X$. 
\begin{definition}
A functor $p \colon \Cat{E} \rightarrow \Cat{C}$ is a \emph{(poset) fibration} if
\begin{itemize} 
\item each fibre $\Cat{E}_X$ is a poset category (that is,
at most one arrow between every two objects); the corresponding
order on objects is denoted by $\leq$;
\item for every $f \colon X \rightarrow Y$ in $\Cat{C}$
	and object $S$ above $Y$ 
	there is a \emph{Cartesian morphism} $\widetilde{f}_S \colon f^*(S) \rightarrow S$
	above $f$, with the property that for 
	every arrow $g \colon Z \rightarrow X$,
	every object $R$ above $Z$ and arrow $h \colon R \rightarrow S$
	above $f \circ g$, 
	there is a unique arrow $k \colon R \rightarrow f^*(S)$ above $g$
	such that $\widetilde{f}_S \circ k = h$. 
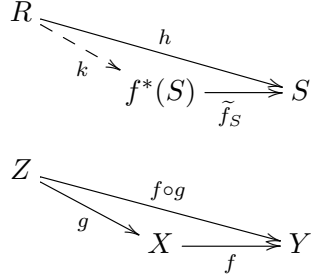
	$$
	\xymatrix@R=0.5cm@C=1cm{
		R \ar@{-->}[dr]_k \ar[drr]^h
			& 
			& \\
			& f^*(S) \ar[r]_{\widetilde{f}_S} 
			& S \\
		Z \ar[dr]_g \ar[drr]^{f \circ g}
			& & \\
			& X \ar[r]_f
			& Y
	}
	$$
\end{itemize}
\end{definition}
\begin{remark}
In this paper we only consider poset fibrations, and refer to them simply as fibrations. 
The usual definition of fibration is more general (e.g.,~\cite{Jacobs:fib}):
normally, fibres are not assumed to be posets.
Poset fibrations have several good properties, mentioned below. In the application
to coinductive predicates, it is customary to work with poset fibrations. 
\end{remark}

For a morphism $f \colon X \rightarrow Y$, the 
assignment $R \mapsto f^*(R)$ gives rise to a 
functor
$f^* \colon \Cat{E}_Y \rightarrow \Cat{E}_X$,
called \emph{reindexing along $f$}. (Note that functors between poset categories are just monotone maps.) 
We use a strengthening of poset fibrations, following~\cite{SprungerKDH18,KKHKH19}.
\begin{definition}
A poset fibration $p \colon \Cat{E} \rightarrow \Cat{C}$ is called a $\CLat$-fibration if 
$(\Cat{E}_X, \leq)$ is a complete lattice for every $X$, and reindexing preserves arbitrary meets. 
\end{definition}

Any poset fibration $p$ is split: we have $(g \circ f)^* = f^* \circ g^*$ for any morphisms $f,g$ that compose. 
Further, $p$ is faithful. This captures the intuition that morphisms in $\Cat{E}$ are morphisms in $\Cat{C}$
with a certain property; e.g., relation-preserving, or non-expansive (Examples~\ref{ex:rel-fib},~\ref{ex:relv-fib}).
We note that $\CLat$-fibrations are instances of topological functors~\cite{herr74:topo}. 
We use the former, in line with existing related work~\cite{HasuoKC18,KKHKH19}.
This also has the advantage of keeping our results 
amenable to possible future extensions to a wider class of examples.

\begin{example}\label{ex:rel-fib}
	Consider the \emph{relation fibration} $p \colon \Rel \rightarrow \Set$, where $p(R \subseteq X \times X) = X$. 
	Reindexing is given by inverse image: for a map $f \colon X \rightarrow Y$
	and a relation $S \subseteq Y \times Y$, we have $f^*(S) = (f \times f)^{-1}(S)$. 
	The functor $p$ is a $\CLat$-fibration. 
	
	Closely related is the \emph{predicate fibration} $p \colon \Pred \rightarrow \Set$.
	An object of $\Pred$ is a pair $(X,\Gamma)$ consisting of a set $X$ and a subset $\Gamma \subseteq X$,
	and an arrow from $(X,\Gamma)$ to $(Y,\Theta)$ is a map $f \colon X \rightarrow Y$ such that $x \in \Gamma$ implies $f(x) \in \Theta$.  
	The functor $p$ is given by $p(X,\Gamma) = X$, reindexing is given by inverse image, and $p$ is a $\CLat$-fibration as well. 
\end{example}	
In the relation fibration, we sometimes refer to an object $(X,R \subseteq X^2)$ simply by $R$, and similarly
in the predicate fibration.

\begin{example}\label{ex:relv-fib}
	Let $\V$ be a complete lattice. Define the category $\RelV$ as follows:
	an object is a pair $(X, d)$ where $X$ is a set and a function 
	$d \colon X \times X \rightarrow \V$, and a morphism from $(X,d)$ to $(Y,e)$
	is a map $f \colon X \rightarrow Y$ such that $d(x_1,x_2) \leq e(f(x_1),f(x_2))$
	for all $x_1,x_2 \in X$. 
	The forgetful functor $p \colon \RelV \rightarrow \Set$ is a $\CLat$-fibration,
	where reindexing along $f \colon X \rightarrow Y$ is given by
	$f^*(Y,e) = (X, e \circ f \times f)$. 
	
	For $\V = 2 = \{0,1\}$ with the usual order $0 \leq 1$, $\RelV$ coincides with
	$\Rel$. Another example is given by the closed interval $\V = [0,1]$,
	with the \emph{reverse} order. Then, 
	a morphism from $(X,d)$ to $(Y,e)$ is a \emph{non-expansive map} $f \colon X \rightarrow Y$,
	that is, s.t.\ $e(f(x_1),f(x_2)) \leq d(x_1,x_2)$ (with $\leq$ the usual order, i.e., where $0$ is the smallest). 
	This instance will be denoted by $\Rel_{[0,1]}$. 
\end{example}

\subsection*{Liftings and Coinductive Predicates}

Let $p \colon \Cat{E} \rightarrow \Cat{C}$ be a fibration, and $B \colon \Cat{C} \rightarrow \Cat{C}$ a functor. 
A functor $\overline{B} \colon \Cat{E} \rightarrow \Cat{E}$ is called a \emph{lifting}
of $B$ if $p \circ \overline{B} = B \circ p$. 
In that case, $\overline{B}$ restricts
to a functor $\overline{B}_X \colon \Cat{E}_X \rightarrow \Cat{E}_{BX}$, for any $X$ in $\Cat{C}$. 

A lifting $\overline{B}$ of $B$ gives rise to an abstract notion of coinductive predicate, as follows. 
For any $B$-coalgebra $(X,\gamma)$ there is the functor, i.e., monotone function defined by
$\gamma^* \circ \overline{B}_X \colon \Cat{E}_X \rightarrow \Cat{E}_X$. 
We think of post-fixed points of $\gamma^* \circ \overline{B}_X$ as \emph{invariants}, generalising
\emph{bisimulations}. 
If $p$ is a $\CLat$-fibration, then $\gamma^* \circ \overline{B}_X$ has a greatest fixed point 
$
\nu(\gamma^* \circ \overline{B}_X)
$, which is also the greatest post-fixed point. It is referred to as the \emph{coinductive predicate} defined by $\overline{B}$ on $\gamma$.

\begin{example}
	First, for a $\Set$ functor $B \colon \Set \rightarrow \Set$, 
	recall the lifting $\Rel(B)$ of $B$ defined in the beginning of this section.
	We refer to $\Rel(B)$ as the \emph{canonical relation lifting} of $B$. 
	For a coalgebra $(X,\gamma)$, a post-fixed point of the operator $\gamma^* \circ \Rel(B)_X$ is a bisimulation,
	as explained above. The coinductive predicate
	$\nu(\gamma^* \circ \Rel(B)_X)$ defined by $\Rel(B)$ is bisimilarity. 
	Another example is given by the lifting $\overline{B}$ for similarity defined in the
	beginning of this section, which we further study in Section~\ref{sec:examples}. 
	In that section we also define a unary predicate, divergence, making use of the predicate fibration. 
	Coinductive predicates in the fibration $\Rel_{[0,1]}$ can be thought of as
	\emph{behavioural distances}, providing a quantitative analogue of bisimulations, 
	measuring the distances between states. A simple example
	on deterministic automata is studied in Section~\ref{sec:sdw-distance}.
\end{example}

\begin{remark}
In quantitative examples one often works in a category with more
structure, e.g., by replacing $\Rel_{[0,1]}$ by the category of pseudo-metrics and non-expansive maps. Similarly, one
can replace $\Rel$ by the category of equivalence relations. Defining liftings then 
requires slightly more work, and since we use fibrations to define coinductive predicates, 
this is not needed. Therefore, we do not use such categories in our examples. 
\end{remark}

We sometimes need the notion of \emph{fibration map}: if $\overline{B}$ is a lifting of $B$, the 
pair $(\overline{B},B)$ is called a fibration map if $(Bf)^* \circ \overline{B}_Y = \overline{B}_X \circ f^*$ for any arrow $f \colon X \rightarrow Y$ in $\Cat{C}$. 
If $B$ preserves weak pullbacks, then $(\Rel(B), B)$ is a fibration map~\cite{jacobs-coalg}
in the relation fibration (Example~\ref{ex:rel-fib}).

\subsection{Coalgebraic Modal Logic}
\label{sec:logic}

We recall a general duality-based approach to coalgebraic modal logic where we work in the context of a 
contravariant adjunction~\cite{PavlovicMW06,Klin07,JacobsS09} in contrast to earlier
work~\cite{KupkeKP04,BonKur05} that assumed a dual equivalence.

We assume the following setting, involving an adjunction $P \dashv Q$ and a natural
transformation $\delta \colon BQ \Rightarrow QL$:
\begin{equation}\label{eq:logic}
\xymatrix{
	\Cat{C}\ar@/^2ex/[rr]^-{P} \save !L(.5) \ar@(dl,ul)^{B} \restore 
	& \bot 
	& \Cat{D}^{\op} \ar@/^2ex/[ll]^-Q \save !R(.5) \ar@(ur,dr)^{L} \restore 
	& \mbox{\qquad with} &
	BQ \ar@{=>}[r]^-{\delta} & QL
}
\end{equation}
In this context, a \emph{logic} for $B$-coalgebras is a pair $(L, \delta)$ as above. The functor
$L \colon \Cat{D} \rightarrow \Cat{D}$ 
represents the syntax of the modalities. It is assumed to have an initial algebra $\alpha \colon L\Phi \stackrel{\cong}{\rightarrow} \Phi$,
which represents the set (or other structure) of formulas of the logic. The natural transformation $\delta$ gives the one-step semantics. It can equivalently be presented in terms of its \emph{mate} 
$\widehat{\delta} \colon LP \Rightarrow PB$, which is perhaps more common in the literature. However,
we will formulate adequacy and expressiveness in terms of the current presentation of $\delta$.

Let $(X,\gamma)$ be a $B$-coalgebra. 
The semantics $\Sem{\_}$ of a logic $(L,\delta)$ arises by initiality of $\alpha$, making use of the mate $\widehat{\delta}$,
as the unique map making the diagram on the left below commute. 
$$
\xymatrix{
L \Phi \ar@{-->}[r]^{L \Sem{\_}} \ar[d]_\alpha & L P X \ar[r]^{\widehat{\delta}} & 
    P B X \ar[d]^{P \gamma} 
	&
	X \ar@{-->}[rr]^-{\theory} \ar[d]_-{\gamma}
		&
		& Q\Phi \ar[d]^-{Q\alpha}
		\\
	\Phi \ar@{-->}[rr]^{\exists ! \Sem{\_}} & & P X 
	&
	BX \ar@{-->}[r]^-{B\theory} 
		& BQ\Phi \ar[r]^-{\delta}
		& QL\Phi 
}
$$
The \emph{theory map} $\theory \colon X \rightarrow Q\Phi$ is defined as the transpose of $\Sem{\_}$, i.e.,
$\theory = Q \Sem{\_} \circ \eta_X$ where $\eta \colon \Id \rightarrow QP$ is the unit of the adjunction $P \dashv Q$.
It is the unique map making the diagram on the right above commute. 

\begin{example}\label{ex:logic-dfa}
	Let $\Cat{C}=\Cat{D}=\Set$, $P=Q=2^-$ the contravariant powerset functor, and $BX = 2 \times X^A$. 
	We define a simple logic for $B$-coalgebras, where formulas are just words over $A$. 
	To this end, let $LX = A \times X + 1$. The initial algebra of $L$ is the set $A^*$ of words. 
	Define $\delta \colon BQ \Rightarrow QL$ on a component $X$ as follows:
	$$
	\delta_X \colon 2 \times (2^X)^A \rightarrow 2^{A \times X + 1} \qquad
	\delta_X(o,t)(u) = 
		\begin{cases}
			o & \text{ if } u = \ast \in 1 \\
			t(a)(x) & \text{ if } u = (a,x)  \in A \times X
		\end{cases}		 
	$$
	For a coalgebra $\langle o, t \rangle \colon X \rightarrow 2 \times X^A$, the associated theory map $\theory \colon X \rightarrow 2^{A^*}$
	is given by $\theory(x)(\varepsilon) = o(x)$ and $\theory(x)(aw) = \theory(t(x)(a))(w)$ for all $x \in X$, $a \in A$, $w \in A^*$. 
	This is, of course, the usual semantics of deterministic automata. 
\end{example}

In the above example, the logic does not contain propositional connectives; this is reflected
by the choice $\Cat{D} =\Set$. 
Although it is possible to include propositional connectives into the functor $L$ (cf.~e.g.~\cite{Klin07}),
one usually adds those connectives by choosing $\Cat{D}$ to be a category of algebras.
For instance,  Boolean algebras are a standard choice for propositional logic,
and in Section~\ref{sec:examples} we use the category of semilattices to represent conjunction. 
In fact, if one is only interested in defining the semantics of the logic, one can simply
work with algebras for a signature; this is supported by the adjunctions presented in the next subsection. 
We outline in the next subsection how this can be used to represent
the propositional part of a real-valued modal logic.

\subsection{Contravariant Adjunctions}
\label{sec:adjunctions}

In this subsection we discuss several adjunctions that we use for 
presenting coalgebraic logic  as above, and will allow us
in Section~\ref{sec:examples} to demonstrate that a large variety of concrete examples is
covered by our framework. 
In all cases,
the adjunctions that we use for
the logic are generated by an object $\Omega$ of `truth values'. In fact, we believe all of the dual adjunctions listed 
in this section are instances of the so-called concrete dualities from~\cite{dualities} where $\Omega$ is 
the dualising object inducing the adjunction. 

For a simple but useful class of such adjunctions, let $\Cat{D}$ be a category with products, and $\Omega$
an object in $\Cat{D}$. 
Then there is an adjunction 
\begin{equation}\label{eq:adjunction-simple}
P \dashv Q \colon \Set \leftrightarrows \Cat{D}^{\op} \qquad \text{ where } 
PX = \Omega^X \text{ and } QX = \Hom(X,\Omega) \,,
\end{equation}
where $\Omega^X$ is the $X$-fold product of $\Omega$. 

This adjunction is instrumental for representing the semantics of a coalgebraic modal logic 
for $B$-coalgebras  based on predicate liftings (cf.~e.g.~\cite{KupkeP11}) within the dual adjunction framework by defining a suitable category of $L$-algebras.    
In general, describing the category of $L$-algebras 
that {\em precisely} represents a given logic (i.e., where the initial algebra corresponds to the set of formulas {\em modulo equivalence}) 
is nontrivial. For studying expressiveness, however, it is sufficient to consider formulas and their semantics.
This can be done as follows:
We start by considering 
a set $\pop$ of
propositional operators,  
each $o \in \pop$ associated with a certain
finite arity $\mathrm{ar}(o) \in \mathbb{N}$ 
and define the (propositional) signature functor
\begin{equation}\label{eq:signature}
\Sigma_\pop \colon \Set \to \Set \qquad \mbox{by putting} \qquad \Sigma_\pop X \coloneqq \coprod_{o \in \pop} X^{\mathrm{ar}(o)}.
\end{equation}
The category $\Alg(\Sigma_\pop)$ of algebras for the functor $\Sigma_\pop$ 
will play the role of the category $\Cat{D}$ in~\eqref{eq:adjunction-simple}.
We assume that we are given a set of truth values $\TV$ together with
a $\Sigma_\pop$-algebra structure $a_\TV \colon \Sigma_\pop \TV \to \TV$, which gives an interpretation
of the propositional operators. 
As $\TV$ is a $\Sigma_\pop$-algebra we obtain  functors $P \colon \Set \to \Alg(\Sigma_\pop)^\op$ and $Q \colon \Alg(\Sigma_\pop)^\op \to \Set$ 
as described in~\eqref{eq:adjunction-simple}.  An $\Omega$-valued coalgebraic modal logic $\logic(\Lambda)$ 
for a functor $B \colon \Set \to \Set$ is now given as  
a set $\Lambda$ of modal operators 
where each
$\lambda \in \Lambda$ is an $\Omega$-valued predicate lifting $\lambda \colon P^n  \Rightarrow PB$
with $\mathrm{ar}(\lambda)=n$ the arity of $\lambda$.
Given  $\logic(\Lambda)$ we define $L_\Lambda \colon \Alg(\Sigma_\pop) \to \Alg(\Sigma_\pop)$ by putting 
$$L_\Lambda A \coloneqq T_{\Sigma_\pop}\left(\{ [\lambda](a_1,\dots,a_n) \mid \lambda \in \Lambda, n = \mathrm{ar}(\lambda), a_j \in A \mbox{ for } 1 \leq j \leq n\}\right)$$
where $[\lambda](a_1,\dots,a_n)$ should be understood as name of a generator and where $ T_{\Sigma_\pop}$ denotes the free (term) monad over $\Sigma_\pop$.
The action of $L_\Lambda$ on a given morphism $f \colon A \to B$ is defined to be the unique $\Alg(\Sigma_\pop)$-morphism extending the
map $[\lambda_i](a_1,\ldots,a_n) \mapsto [\lambda_i](f(a_1),\ldots,f(a_n))$. 
It is now easy to see that the predicate liftings in $\Lambda$ give rise to a natural transformation
$\widehat{\delta} \colon L_\Lambda P \Rightarrow P B$ where, for an arbitrary set $X$, the $X$ component 
$\widehat{\delta}_X \colon L_\Lambda PX   \to P B X$ is the unique extension of the map
\[ [\lambda](u_1,\ldots,u_n) \mapsto  \lambda(u_1,\ldots,u_n) \] 
for $\lambda \in \Lambda$, $n = \mathrm{ar}(\lambda)$ and $u_j \in\Omega^X$ for $1 \leq j \leq n$. 
In other words, $(L_\Lambda, \delta)$ with $\delta$ being the mate of $\widehat{\delta}$ is a
logic for $B$-coalgebras in the sense of~\eqref{eq:logic}.
We arrive at the following picture:
   \begin{equation}\label{eq:comolo} 
   \xymatrix{
   \Set\ar@/^1.5ex/[rr]^-{P = \Omega^{-}} \save !L(.5) \ar@(dl,ul)^{B} \restore 
& \bot 
& \Alg(\Sigma_\pop)^{\op} \ar@/^2ex/[ll]^-{Q = \Hom(-,\Omega)}  \save !R(.8) \ar@(ur,dr)^{L_\Lambda} \restore 
} \hspace{1.5cm} \delta \colon B Q \Rightarrow Q L_\Lambda 
\end{equation}

\begin{example}
     To illustrate the 
     outlined approach, consider the real-valued coalgebraic modal logics from~\cite{KonigM18}. 
     The set $\Phi$ of formulas of these logics is given by the following definition that is indexed by a set $\Lambda$ of unary modal operators:
      \[  \Phi ::= \top \mid [\lambda]\varphi, \lambda \in \Lambda \mid \min(\varphi_1,\varphi_2) \mid \neg \varphi \mid \varphi \ominus q, q \in \mathbb{Q} \cap [0,\top] \]
      where $[0, \top]$ is a closed interval of real numbers with $\top$ denoting an arbitrary positive real number,
    $\ominus$ is interpreted as truncated subtraction on $[0,\top]$ given by $p \ominus q \mathrel{:=} \max(p-q,0)$, 
      $\min$ is interpreted as minimum and negation on $[0,\top]$ is defined as $\neg q \mathrel{:=} \top - q$.
      Following the construction of $L_\Lambda$ as described above, we obtain the following dual adjunction: 
   \[\xymatrix{
   \Set\ar@/^1.5ex/[rr]^-{P = [0,\top]^{-}} \save !L(.5) \ar@(dl,ul)^{B} \restore 
& \bot 
& \Alg(\Sigma_{[0,\top]})^{\op} \ar@/^2ex/[ll]^-{Q = \Hom(-,[0,\top])}  \save !R(.8) \ar@(ur,dr)^{L_\Lambda} \restore 
}  .\] 
      Here the operations on $[0,\top]$ are $\top$, $\min$, $\neg$ and $- \ominus q$ for $q \in \mathbb{Q} \cap [0,\top]$, 
      thus 
      $$\Sigma_{[0,\top]} X = 1 + X^2 + X + X \times (\mathbb{Q} \cap [0,\top]) \quad \mbox{ and } \quad 
      L_\Lambda(A) = T_{\Sigma_{[0,\top]}} (\{ [\lambda] a \mid a \in A, \lambda \in \Lambda \}).$$ 
\end{example}

To study expressiveness relative to a coinductive predicate in a fibration $p \colon \Cat{E} \to \Cat{C}$
we rely on a given dual adjunction $P \dashv Q$ between $\Cat{C}$ and $\Cat{D}$ together with its lifted version
$\lift{P} \dashv \lift{Q}$ between $\Cat{E}$ and $\Cat{D}$. In a large class of examples the fibration under consideration 
will be of type $p \colon \RelV \to \Set$ with $P \dashv Q$ being the dual adjunction between $\Set$ and $\Alg(\Sigma)$ 
described above.  
We will now provide a proposition that yields the required dual adjunction $\lift{P} \dashv \lift{Q}$ between $\RelV$ and
$\Alg(\Sigma)$.
To obtain this dual adjunction 
we need a number of assumptions. First we make some assumptions
on the truth and distance values $\TV$ and $\V$:
\begin{itemize}
 \item $\V$ is a complete lattice of distance values,
 \item $\TV$ is a bounded poset of truth values, 
 \item $(\TV,R_\TV \colon \TV \times \TV \to \V) \in \RelV$.
\end{itemize}
Furthermore we let $\Eq \colon \Set \rightarrow \RelV$
be the diagonal functor given by $\Eq X = \Delta_X$ where
\begin{eqnarray*} 
        \Delta_X(x_1,x_2) & \mathrel{:=} & \left\{ \begin{array}{ll} 
                                                    \top & \mbox{ if } x_1 = x_2 \\ 
                                                    \perp & \mbox { otherwise.}
                                                   \end{array} \right.
\end{eqnarray*}

\begin{proposition}\label{prop:tool}
 Let $\TV$ and $\V$ be sets of truth and distance values that satisfy the above assumptions
 and let $\Sigma \colon \Set \to \Set$ be a functor. 
 Suppose furthermore that $\Sigma$ has a lifting $\lift{\Sigma}\colon \RelV \to \RelV$ such that
 (i) $\Eq \circ \Sigma  \leq \lift{\Sigma} \circ \Eq$ and (ii)
  for any $(X,R), (Y,S) \in \RelV$ there is a morphism 
       $\overline{\strength}_{R,S} \colon R \times \lift{\Sigma} S \to \lift{\Sigma}(R \times S)$
       above the strength map $\strength_{X,Y} \colon X \times \Sigma Y \to \Sigma (X \times Y)$
       (the latter exists for any set functor $\Sigma$).
   If  $\lift{\Sigma} R_\TV \leq a_\TV^*(R_\TV)$, then
   there is a dual adjunction
   \begin{equation}\label{eq:dual-adjunction}
  \xymatrix{
  \RelV \ar@/^2ex/[rr]^-{\Hom(\_,R_\TV )} 
  & \bot & \Alg(\Sigma)^{\mathsf{op}} \ar@/^2ex/[ll]^-{\Hom(\_,a_\TV)} 
  }
  \end{equation}
\end{proposition}
\begin{proof}
    We first have to show that the functors that form the adjunction are well-defined. 
      In the following we write $\alpha$ as abbreviation for an algebra $(A,\alpha)$. 
    Throughout this proof we denote the least and the largest element  of $\V$ by $\perp$ and $\top$, respectively.
    Recall that the condition for a function $f$ to be a $\RelV$-morphism is $R_1(x_1,x_2) \leq R_2(f x_1,f x_2)$.
 
    To see that the functors are well-defined on objects, first note that for each $\Sigma$-algebra $(A,\alpha)$, the set
    $\Hom(\alpha,a_\TV)$ can be turned into a $\RelV$-object by defining
    \begin{eqnarray*}
      R_{\Hom(\alpha,a_\TV)}\colon \Hom(\alpha,a_\TV) \times \Hom(\alpha,a_\TV) & \to & \V \\
      (h_1,h_2) & \mapsto & \bigwedge_{a \in A} R_\TV(h_1(a),h_2(a)) .
    \end{eqnarray*}
    
    Likewise, for each $\RelV$-object $(X,R)$, the set $ \Hom(R,R_\TV)$ 
    carries a $\Sigma$-algebra structure  
   $a_R\colon \Sigma \Hom(R,R_\TV) \to \Hom(R,R_\TV)$
    given by  the function that maps $t\colon 1 \to  \Sigma \Hom(R,R_\TV)$
    to the following composition of arrows:
\[
\xymatrix{R \cong R \times 1 \ar[r]^-{\id \times t} & R \times \Eq \Sigma \Dis \ar[r]^{\leq} & R \times \lift{\Sigma}  \Eq  \Dis   \ar[r]^-{\lift{\strength}_{R,\Eq \Dis}} & \lift{\Sigma}(R \times \Eq \Dis)
\ar[r]^-{\lift{\Sigma} \mathrm{ev}} & \lift{\Sigma}R_\TV \ar[r]^{\lift{a_\TV}} & R_\TV }
\]
where $\Dis =  \Hom(R,R_\TV)$
and where $\mathrm{ev}(x,f) = f(x)$ is the evaluation function. 
To see that the above is a well-defined $\Rel_V$-morphism we only have to check that $\mathrm{ev} \in \RelV$ as the other
arrows are morphisms by our assumption that $\Sigma$ has a lifting $\lift{\Sigma}\colon \RelV \to \RelV$ such that $\strength$ and $a_\TV$
become morphisms in $\RelV$. We now show that $\mathrm{ev}$ satisfies the $\RelV$ morphism condition. Consider
two pairs $(x_1,f_1),(x_2,f_2) \in X \times \Hom(R,R_\TV)$. We distinguish cases:
\begin{description}
  \item[Case] $f_1 = f_2 = f$ In this case we have
     \begin{eqnarray*} R \times \Eq\Dis ((x_1,f),(x_2,f))  & = &  R(x_1,x_2) \wedge \top  = R(x_1,x_2) \\
       & \leq & R_\TV (f(x_1),f(x_2)) = R_\TV (\mathrm{ev}(x_1,f),\mathrm{ev}(x_2,f)). 
     \end{eqnarray*}
  \item[Case] $f_1 \not= f_2$. Then
  \[ R \times \Eq \Dis ((x_1,f_1),(x_2,f_2)) = R(x_1,x_2) \wedge \perp  \; = \; \perp \;\leq \; R_\TV (\mathrm{ev}(x_1,f_1),\mathrm{ev}(x_2,f_2)). \] 
\end{description}

To see that the $\Hom$-functors are well-defined on morphisms we first check that $\Hom(\_,a_\TV)$ maps algebra morphisms
to morphisms in $\RelV$.  To this aim consider an algebra morphism $h\colon (A_1,\alpha_1) \to (A_2,\alpha_2)$ and 
$g_1,g_2 \in \Hom(\alpha_2,a_\TV)$. We calculate:
\begin{eqnarray*}
        R_{\Hom(\alpha_2,a_\TV)}(g_1,g_2) & = &  \bigwedge_{a \in A_2} R_\TV(g_1(a),g_2(a)) \leq \bigwedge_{a' \in A_1} R_\TV(g_1(h(a')),g_2(h(a'))) \\
             & = &   R_{\Hom(\alpha_1,a_\TV)}\left(\Hom(h,a_\TV)(g_1),\Hom(h,a_\TV)(g_2)\right)
\end{eqnarray*}
We now check that the functor $\Hom(\_,R_\TV)$ is well-defined on morphisms as well. Let $h\colon (X_1,R_1) \to (X_2,R_2) \in \RelV$.
We have to show that $\Hom(h,R_\TV)\colon  \Hom(R_2,R_\TV) \to \Hom(R_1,R_\TV)$ is a $\Sigma$-algebra morphism. 

We calculate:
	\begin{eqnarray*}
		\Hom(h,R_\TV) (a_{\Hom(R_2,R_\TV)}(t)) & = & \Hom(h,R_\TV) (\lambda x. a_\TV \circ \Sigma \mathrm{ev} \circ \strength (x,t))\\
		 & = & \lambda y. a_\TV \circ \Sigma \mathrm{ev} \circ \strength (h(y),t) \\
		 & \stackrel{\mbox{\tiny nat. of $\strength$}}{=} &  \lambda y. a_\TV \circ \Sigma \mathrm{ev} \circ \Sigma (h \times \id) \circ  \strength (y,t) \\
		 & \stackrel{\mbox{\tiny (*)}}{=} &  \lambda y. a_\TV \circ \Sigma \mathrm{ev} \circ \Sigma (\id \times \Hom(h,R_\TV)) \circ  \strength (y,t) \\
		 &  \stackrel{\mbox{\tiny nat. of $\strength$}}{=} & \lambda y. a_\TV \circ \Sigma \mathrm{ev}  \circ  \strength (y,\Sigma\Hom(h,R_\TV)(t))\\
		 &   \stackrel{\mbox{\tiny Def.~of $a_{\Hom(\dots)}$}}{=}& a_{\Hom(R_1,R_\TV)}(\Sigma\Hom(h,R_\TV)(t))
	\end{eqnarray*}
where (*) holds as the following diagram can be easily seen to commute in $\Set$: 
$$
   \xymatrix{X_1 \times \Hom(X_2,\TV) \ar[d]_{h \times \id} \ar[rr]^-{\id \times \Hom(h,\TV)} &  & X_1 \times \Hom(X_1,\TV) \ar[d]^-{\mathrm{ev}}  \\  
    X_2 \times \Hom(X_2,\TV) \ar[rr]_-{\mathrm{ev}} & & \TV    } 
$$

This finishes the argument that the functors are well-defined. We will now argue that they form an adjunction.
We prove this by defining the unit and counit of the adjunction satisfying the triangle identities. 

For $(X,R) \in \RelV$ we define the unit map $\eta_R\colon  R \to \Hom( \Hom(R,R_\TV), a_\TV)$
by putting $\eta_R(x) \mathrel{:=} \lambda f. f(x)$. Naturality of $\eta$ can be easily checked (left to the reader), but well-definedness
is not obvious. For the latter we have to show that $\eta_R$  is a $\RelV$-morphism and that $\eta_R(x)$ is a $\Alg(\Sigma)$-morphism for all
$(X,R) \in \RelV$ and all $x \in X$. 

To see that $\eta_R$ is a  $\RelV$-morphism, consider $x_1,x_2 \in X$. 
\begin{eqnarray*} R(x_1,x_2) & \leq & \bigwedge_{f \in \Hom(R,R_\TV)} R_\TV(f(x_1),f(x_2))=  \bigwedge_{f \in \Hom(R,R_\TV)}  R_\TV(\eta_R(x_1)(f),\eta_R(x_2)(f)) \\
        & = & R_{\Hom(\Hom(R,R_\TV),a_\TV)}(\eta_R(x_1),\eta_R(x_2))
\end{eqnarray*}
To check that $\eta_R(x)$ is a $\Sigma$-algebra morphism, we calculate
\begin{eqnarray*}
    (a_\TV \circ \Sigma \eta_R(x))(t) & = & (a_\TV \circ \Sigma (\lambda f. f(x))(t) =  (a_\TV \circ \Sigma (\mathrm{ev}(x,\_)))(t) \\
    & = & (a_\TV \circ \Sigma (\mathrm{ev} \circ \strength))(x,t) \\
    & = & \eta_R(x) \left(\lambda x.(a_\TV \circ \Sigma (\mathrm{ev} \circ \strength))(x,t)  \right)  =  \eta_R(x) \left( a_{\Hom(R,R_\TV)}(t) \right) 
\end{eqnarray*}

For the counit of the (dual) adjunction we define $\epsilon_\alpha\colon  (A,\alpha) \to \Hom(\Hom(\alpha,a_\TV),R_\TV)$ by putting
$\epsilon_\alpha(a) \mathrel{:=} \lambda g.g(a)$ for all $(A,\alpha) \in \Alg(\Sigma)$ and all $a \in A$.
Again we leave it to the reader to convince themselves that $\epsilon$ is natural. We have to check well-definedness, i.e, we need to check
that $\epsilon_\alpha(a)$ is a $\RelV$-morphism and that $\epsilon_\alpha$ is an  $\Alg(\Sigma)$-morphism.

To see that $\epsilon_\alpha(a)$ is a $\RelV$-morphism we consider $g_1,g_2 \in \Hom(\alpha,a_\TV)$:
\begin{eqnarray*}
    R_{\Hom(\alpha,a_\TV)}(g_1,g_2) & = & \bigwedge_{a' \in A} R_\TV (g_1(a'),g_2(a')) \leq  R_\TV (g_1(a),g_2(a)) \\
        & = & R_\TV(\epsilon_\alpha(a)(g_1),\epsilon_\alpha(a)(g_2))
\end{eqnarray*}
To check that $\epsilon_\alpha$ is an $\Alg(\Sigma)$-morphism we calculate:
\begin{eqnarray*}
     a_{\Hom(\Hom(\alpha,a_\TV),R_\TV)} \circ \Sigma \epsilon_\alpha (t) & = & (\lambda g. (a_\TV \circ  \Sigma \mathrm{ev} \circ \strength)(g,\Sigma \epsilon_\alpha(t)) \\
     & \stackrel{\mbox{\tiny Def. of $\strength$}}{=} &  (\lambda g. (a_\TV \circ  \Sigma \mathrm{ev} (g,\_))) (\Sigma \epsilon_\alpha(t)) \\
     & \stackrel{\mbox{\tiny (+)}}{=} &  (\lambda g. (a_\TV \circ  \Sigma g)) (t) \\
      & \stackrel{\mbox{\tiny $g$ alg.~mor.}}{=} & \lambda g. g (\alpha(t)) = \epsilon_\alpha(\alpha(t)) 
\end{eqnarray*}
where (+) is an easy consequence of $(\mathrm{ev}(g,\_) \circ \epsilon_\alpha)(a) = g(a)$ for all $a \in A$.

This finishes the definition of unit and counit of the adjunction - checking  the triangle equalities is a straightforward exercise.
\end{proof}

The following remark is obvious, but at the same time useful for concrete examples. 
\begin{remark}\label{rem:adjunction}
   Let $\mathcal{C}$ be a full subcategory of $\RelV$ and $\mathcal{D}$ a full subcategory of $\Alg(\Sigma)$ 
   such that  $\Hom(-,a_\TV)$ and $\Hom(-,R_\TV)$ restrict to functors of type $\mathcal{D} \to \mathcal{C}$ and
   of type $\mathcal{C} \to \mathcal{D}$, respectively. 
   Then the dual adjunction from Prop.~\ref{prop:tool} restricts to a dual adjunction between $\mathcal{C}$ and $\mathcal{D}$.
\end{remark}

The assumptions in Proposition~\ref{prop:tool} concerning existence of a suitable lifting of $\Sigma$  are in particular met
when $\Sigma$  is a polynomial functor.

\begin{corollary}\label{cor:polynomial}
    Let $\TV$ and $\V$ be sets of truth and distance values that satisfy the assumptions from~Prop.~\ref{prop:tool}, 
    let $\Sigma$ be a signature functor. 
   Then $\Sigma$ lifts to $\RelV$ such that  (i) $\Eq \circ \Sigma  \leq \lift{\Sigma} \circ \Eq$ and (ii) there is a strength map $\lift{\strength}_{R,S}$ above $\strength_{X,Y}$.
   Consequently, the dual adjunction $\Hom(\_,R_\TV ) \dashv \Hom(\_,a_\TV)$ from~(\ref{eq:dual-adjunction}) $\Hom(\_,R_\TV ) \dashv \Hom(a_\TV,\_ )$ exists 
   if $\lift{\Sigma} R_\TV \leq a_\TV^*(R_\TV)$.
\end{corollary}
\begin{proof}
    It is clear that the existence of the dual adjunction follows from Prop~\ref{prop:tool} once we establish that
    any polynomial functor $\Sigma$ has a lifting to $\RelV$  such that  (i) $\Eq \circ \Sigma  \leq \lift{\Sigma} \circ \Eq$  and
    (ii) there is a strength map $\lift{\strength}_{R,S} \colon R \times \lift{\Sigma} S \to \lift{\Sigma}(R \times S)$ above $\strength_{X,Y}$
    for all  $R \in {(\RelV)}_X, S\in {(\RelV)}_Y$. 
    In the following we prove this claim not only for signature functors but for the collection of functors $\mathcal{F}$ generated 
    by the following grammar:
   \[  \mathcal{F} \mathrel{::=} A \in \Set \mid \Id \mid \prod_{j \in J} \Sigma_j \mid \coprod_{j \in J} \Sigma_j \]
   where $J$ are arbitrary sets of indexes, $A$ denotes the constant functor and $\Id \colon \Set \to \Set$ denotes the identity functor.
   For a functor $\Sigma \in \mathcal{F}$ we now inductively define the action of its lifting $\lift{\Sigma} \colon \RelV \to \RelV$ on objects 
   while at the same
   time proving conditions (i) and (ii).
   \begin{description}
    \item[Case] $\Sigma = A$ (constant functor). Then we put
    \[ (\lift{A} R) (a_1,a_2) \coloneqq \left\{ \begin{array}{l} 
                                              \top \qquad \mbox{if } a_1 = a_2 \\
                                              \perp \qquad \mbox{otherwise.}
                                             \end{array} \right. \]
    The conditions on $\Eq$ and  $\strength$ are easy to check as in this case
    $\Eq A X = \Delta_A = \lift{A} \Eq X$ and as
    $\strength_{X,Y}(x,a) = a$ which clearly lifts to a suitable $\lift{\strength}_{R,S}$.
    \item[Case] $\Sigma = \Id$. Then $\lift{\Sigma}(R) = R$, $\Delta \circ \Sigma = \Delta = \lift{\Sigma} \circ \Delta$ and the strength map is simply the identity.
    \item[Case] $\Sigma =  \prod_{j \in J} \Sigma_j$. Then
    \[  \lift{\Sigma}(R)(x_1,x_2) = \bigwedge_{j \in J} \lift{\Sigma}_j(R)(\pi_j(x_1),\pi_j(x_2)) \] 
    where $\pi_j$ is the projection onto the $j$-th component of the product. 
    For proving property (i) we consider an arbitrary set $X$, $x_1,x_2 \in  \Sigma X$ and we calculate:
    \begin{eqnarray*}
     \Eq \Sigma X (x_1,x_2) & = & \Eq \prod_{j \in J} \Sigma_j X (x_1,x_2)
                                  =\bigwedge_{j \in
                                  J} \Eq  (\Sigma_j X) (\pi_j(x_1),\pi_j(x_2))
      \\
      & \stackrel{\mbox{\tiny I.H.}}{\leq}& \bigwedge_{j \in J} \lift{\Sigma}_j (\Eq X) (\pi_j(x_1),\pi_j(x_2)) 
     = \lift{\Sigma} (\Eq X) (x_1,x_2) 
    \end{eqnarray*}
    To check that $\strength \colon X \times \Sigma Y \to \Sigma (X \times Y)$ satisfies condition (ii) 
    let $\strength_j \colon X \times  \Sigma_j Y \to \Sigma_j (X \times Y)$ and consider
    two arbitrary pairs $(x_1,z_1),(x_2,z_2) \in X \times \Sigma Y$. We calculate: 
    \begin{eqnarray*} 
       (R \times \lift{\Sigma}S)((x_1,z_1),(x_2,z_2)) & = &  R(x_1,x_2) \wedge \lift{\Sigma}S(z_1,z_2) \\
       & = & \bigwedge \{ R(x_1,x_2) \} \cup \{\lift{\Sigma}_j S (\pi_j(z_1),\pi_j(z_2)) \mid j \in J \} \\
       & = &  \bigwedge_{j \in J} \left\{(R \times \lift{\Sigma}_j S)((x_1,\pi_j(z_1)),(x_2,\pi_j(z_2)))  \right\}\\ 
       & \stackrel{\mbox{\tiny I.H.}}{\leq} &  \bigwedge_{j \in J} \left\{  \lift{\Sigma}_j (R \times S) (\strength_j(x_1,\pi_j(z_1)),\strength_j(x_2,\pi_j(z_2)))  \right\} \\
       &  \stackrel{\mbox{\tiny (+)}}{=} &   \bigwedge_{j \in J}  \left\{ \lift{\Sigma}_j (R \times S) (\pi_j (\strength (x_1,z_1)),\pi_j (\strength(x_2,z_2))) \right\} \\
       &   \stackrel{\mbox{\tiny Def.}}{=} & \lift{\Sigma}(R \times S)(\strength (x_1,z_1),\strength(x_2,z_2))
    \end{eqnarray*}
    where for (+) we used that $\pi_j \circ \strength = \strength_j \circ (\id \times \pi_j)$ as can be easily checked.
    \item[Case] $\Sigma =  \coprod_{j \in J} \Sigma_j$. Then
    \[  \lift{\Sigma}(R)(\kappa_m(x_1),\kappa_n(x_2)) = \left\{ \begin{array}{l}
                                                                 \lift{\Sigma}_n(R)(x_1,x_2) \mbox{ if } n = m \\
                                                                 \perp \mbox{ otherwise.}
                                                                \end{array}
    \right.  \] 
    where the $\kappa_n$ denotes the $n$-th inclusion into the coproduct. 
    As in the previous case we first verify (i): let $X$ be a set  and consider $x_1,x_2 \in \Sigma X$. W.l.o.g. we assume there are $j \in J$ and   
    $x_1',x_2' \in \Sigma_j X$ with $x_i = \kappa_j(x_i')$ for $i=1,2$ - otherwise property (i) is trivially satisfied.
    Spelling out the definitions we get
    \begin{eqnarray*}
           \Delta(\Sigma X) (x_1,x_2) & = &   \Delta (\Sigma_j X) (x_1',x_2')  \stackrel{\mbox{\tiny I.H.}}{\leq} \lift{\Sigma}_j (\Delta X)  (x_1',x_2')  =  \lift{\Sigma} (\Delta X)  (x_1,x_2)
    \end{eqnarray*}
    Let $\strength_j \colon  X \times  \Sigma_j Y \to \Sigma_j (X \times Y)$
    be the strength maps of the components of $\Sigma$. Consider pairs $(x_1,\kappa_j(y_1)), (x_2,\kappa_j(y_2)) \in X \times \Sigma Y$
    where we assumed that the $y_i$'s are from the same $j$-th component of $\Sigma Y$  - otherwise the strength condition is trivially true.
    We calculate:
    \begin{eqnarray*}
        (R \times \lift{\Sigma}S)((x_1,\kappa_j(y_1)), (x_2,\kappa_j(y_2))) & \stackrel{\mbox{\tiny Def of $\lift{\Sigma}$}}{=} & (R \times \lift{\Sigma}_j S)((x_1,y_1), (x_2,y_2)) \\
        & \stackrel{\mbox{\tiny I.H.}}{\leq} & \lift{\Sigma}_j (\strength_j(x_1,y_1),\strength_j(x_2,y_2)) \\
        & = & \lift{\Sigma}(R\times S) (\kappa_j(\strength_j(x_1,y_1)),\kappa_j(\strength_j(x_2,y_2))) \\
        & = & \lift{\Sigma}(R \times S) (\strength(x_1,\kappa_j (y_1)),\strength(x_2,\kappa_j (y_2))) 
    \end{eqnarray*}
    where the last equality follows from the easily verifiable fact that $\kappa_j \circ \strength_j = \strength \circ (\id \times \kappa_j)$.
   \end{description}
   This finishes the definition of $\lift{\Sigma}$ on objects. Our argument also shows that  for polynomial functors $\Sigma$, the map $\strength$ lifts to $\RelV$
   as required. Finally, we extend $\lift{\Sigma}$ to a functor 
   $\RelV \to \RelV$ by putting $\lift{\Sigma} f \coloneqq
   \Sigma f$  for all morphisms $f: R\to S \in \RelV$. In order to see that $\lift{\Sigma}$ is well defined on morphisms one has to prove that 
   $\Sigma f$ is a $\RelV$-morphism from $\lift{\Sigma}(R)$ to $\lift{\Sigma}(S)$ whenever
   $f:R \to S$ is a $\RelV$-morphism. This can be easily shown by induction on the structure
   of $\Sigma$. Functoriality of $\lift{\Sigma}$ is an immediate consequence of functoriality of
   $\Sigma$.
\end{proof}

Let $\logic(\Lambda)$ be a coalgebraic modal logic for some functor $B$ 
and its representation via a dual adjunction from $\Set$ to $\Alg(\Sigma_\pop)$ for some polynomial functor $\Sigma_\pop$ together with a 
functor $L_\Lambda \colon \Alg(\Sigma_\pop) \to \Alg(\Sigma_\pop)$ as in~\eqref{eq:comolo}.
Furthermore let $a_\TV \colon \Sigma_\pop \TV \to \TV$ be the algebra structure induced by 
the propositional operators $\pop$ of $\logic(\Lambda)$ such that for all operations 
$o \in O$ and all $v_1,\dots,v_n,v_1',\dots,v_n' \in \Omega$ we have
\[ \bigwedge_{1 \leq i \leq n} R_\TV(v_i,v_i')  \leq R_\TV\left(o(v_1,\dots,v_n),o(v_1',\dots,v_n')\right).  \]
Then the results of this section can be summarised in the following diagram:

\begin{center}
$$
	\xymatrix{\RelV \save !L(.5) \ar@(dl,ul)^{\lift{B}} \restore   \ar[dd]_{p} \ar@/^2ex/[rrdd]^-{\overline{P}}  \\ 
	\\
\Set \ar@/^2ex/[rr]^-{P} \save !L(.5) \ar@(dl,ul)^{B} \restore & \bot & \Alg(\Sigma_\TV)^{\op} \ar@/^2ex/[ll]^-Q \ar@/^2ex/[lluu]^-{\lift{Q}} \save !R(.8) \ar@(ur,dr)^{L_\Lambda} \restore 
& \qquad \mbox{ with} &
BQ \ar@{=>}[r]^-{\delta} & QL_\Lambda
}
$$
\end{center}

In the next section we will see that adequacy of the logic $\logic(\Lambda)$ follows if
$\delta$ lifts to $\lift{\delta} \colon \lift{B}\, \lift{Q} \Rightarrow \lift{Q} L_\Lambda$, while expressiveness 
is implied by an additional property of $\lift{\delta}$.

\section{Abstract Framework: Adequacy \& Expressiveness}
\label{sec:framework}

In this section, we define when a logic is adequate and expressive with respect to a coinductive predicate,
and provide sufficient conditions on the logic. 
Coinductive predicates are expressed abstractly via fibrations and functor lifting, 
and logic via a contravariant adjunction. Therefore, we make the following assumptions.  
\begin{assumption}\label{as:basic}
 Throughout this section, we assume:
\begin{enumerate}
	\item \emph{(Type of coalgebra)} An endofunctor $B \colon \Cat{C} \rightarrow \Cat{C}$ on a category $\Cat{C}$;
	\item \emph{(Coinductive predicate)} A $\CLat$-fibration $p \colon \Cat{E} \rightarrow \Cat{C}$ and a lifting $\lift{B} \colon \Cat{E} \rightarrow \Cat{E}$
	of $B$; 
	\item \emph{(Coalgebraic logic)} An adjunction $P \dashv Q \colon \Cat{C} \leftrightarrows \Cat{D}^{\op}$, a functor $L \colon \Cat{D} \rightarrow \Cat{D}$
	with an initial algebra $\alpha \colon L(\Phi) \stackrel{\cong}{\rightarrow} \Phi$, and a natural transformation $\delta \colon BQ \Rightarrow QL$. 
\end{enumerate}
\end{assumption}

As explained in the introduction, to formulate 
adequacy and expressiveness, we need one more crucial ingredient: an object that stipulates
how collections of formulas should be compared. 
In the abstract fibrational setting, we assume an object above $Q\Phi$; 
more systematically, a functor $\lift{Q}$ above $Q$.

\begin{definition}[Adequacy and Expressiveness]
\label{def:adequate-expressive}
Let $\lift{Q} \colon \Cat{D}^{\op} \rightarrow \Cat{E}$ be a functor such that $p \circ \lift{Q} = Q$. 
We say the logic $(L,\delta)$ is
\begin{itemize}
	\item \emph{adequate} if $\nu(\gamma^* \circ \lift{B}_X) \leq \theory^*(\lift{Q} \Phi)$ for every $B$-coalgebra $(X,\gamma)$;
	\item \emph{expressive} if $\nu(\gamma^* \circ \lift{B}_X) \geq \theory^*(\lift{Q} \Phi)$ for every $B$-coalgebra $(X,\gamma)$.
\end{itemize}
\end{definition}
When we need to refer to the functors $\lift{Q}$ or $\lift{B}$ explicitly, we speak about adequacy and expressiveness \emph{via $\lift{Q}$} \emph{w.r.t.\ $\lift{B}$}.
Examples follow in Section~\ref{sec:canonical}, where classical  
expressiveness and adequacy w.r.t.\ bisimilarity is recovered, and Section~\ref{sec:examples}, where
other instances are treated. 

\begin{remark}
	Definition~\ref{def:adequate-expressive} can be generalised to arbitrary poset fibrations, 
	not necessarily assuming complete lattice structure on the fibres,
	as follows. Adequacy means that for any $B$-coalgebra $(X,\gamma)$,
	if $R \leq \gamma^* \circ \lift{B}_X(R)$ then $R \leq \theory^* (\lift{Q} \Phi)$.
	Expressiveness means that for any $B$-coalgebra $(X,\gamma)$,
	we have $\theory^* (\lift{Q} \Phi) \leq R$ for some $R$ with $R \leq \gamma^* \circ \lift{B}_X(R)$. 
	In fact, with these definitions, if $(L, \delta)$ is both adequate and expressive
	then $\gamma^* \circ \lift{B}_X$ has a greatest fixed point, given by
	$\theory^*(\lift{Q} \Phi)$. We prefer to work with $\CLat$-fibrations, since the
	definition is slightly simpler, and it covers all our examples. 
 \end{remark}

\subsection{Sufficient conditions for expressiveness and adequacy}

The results below give conditions on $\overline{B}$, $\overline{Q}$ and primarily 
the one-step semantics $\delta$ that guarantee expressiveness (Theorem~\ref{thm:expressive})
and adequacy (Theorem~\ref{thm:adequate}). 
For simplicity we fix the functor $\lift{Q}$. 
\begin{assumption}\label{as:qbar}
In the remainder of this section 
we assume a functor $\lift{Q} \colon \Cat{D}^{\op} \rightarrow \Cat{E}$ such that $p \circ \lift{Q} = Q$. 
\end{assumption}

For adequacy, the main idea is to require sufficient conditions to lift $\delta$
to a logic for $\overline{B}$. 
\begin{theorem}\label{thm:adequate}
	Suppose that
	\begin{enumerate}
		\item $\lift{B}\,\lift{Q} X \leq \delta_X^*(\lift{Q}L X)$ for every object $X$ in $\Cat{D}$, and
		\item the functor $\lift{Q}$ has a left adjoint.
	\end{enumerate}
	Then $(L,\delta)$ is adequate.
\end{theorem}
\begin{proof} 
The first assumption yields a natural transformation $\lift{\delta} \colon \lift{B} \, \lift{Q} \Rightarrow \lift{Q} L$,
defined on a component $X$ by 
$$
	\lift{\delta}_X = \left(
	\xymatrix{
		\lift{B} \, \lift{Q} X \ar[r]
			& \delta_X^*(\lift{Q} L X) \ar[r]^-{\widetilde{\delta}}
			& \lift{Q} L X
	}
	\right)
$$
where the left arrow is the inclusion $\lift{B}\,\lift{Q} X \leq \delta_X^*(\lift{Q}L X)$,
and the right arrow $\widetilde{\delta}$ is the Cartesian morphism to $\lift{Q}LX$ above $\delta_X$. 
It follows that $\lift{\delta}_X$ is above $\delta_X$. Further, naturality follows from $p$ being 
faithful (as it is a poset fibration, see Section~\ref{sec:coind-pred}) and naturality of $\delta$. 
Observe that we have thus established $(L, \overline{\delta})$ as a logic for $\overline{B}$-coalgebras, 
via the adjunction $\lift{P} \dashv \lift{Q}$. 

Now let $(X,\gamma)$ be a $B$-coalgebra, and $R = \nu(\gamma^* \circ \lift{B}_X)$. Then, in particular, 
$R \leq \gamma^* \circ \overline{B}_X(R)$,
which is equivalent to a coalgebra $\lift{\gamma} \colon R \rightarrow \lift{B} R$ above $\gamma \colon X \rightarrow BX$. 
The logic $(L, \lift{\delta})$ gives us a theory map $\lift{\theory}$ of $(R,\lift{\gamma})$ as the unique
map making the following diagram commute. 
$$
\xymatrix{
	R \ar@{-->}[rr]^-{\lift{\theory}} \ar[d]_-{\lift{\gamma}}
		&
		& \lift{Q}\Phi \ar[d]^-{\lift{Q}\alpha}
		\\
	\lift{B}R \ar@{-->}[r]^-{\lift{B} \, \lift{\theory}} 
		& \lift{B}\,\lift{Q}\Phi \ar[r]^-{\lift{\delta}}
		& \lift{Q}L\Phi 
}
$$
Since $p \circ \lift{Q} = Q$ and $p(\lift{\delta}_\Phi) = \delta_\Phi$, it follows that $p(\lift{\theory})$ 
equals the theory map $\theory$ of $(X,\gamma)$. Hence $R \leq \theory^*(\overline{Q}\Phi)$ as required. 
\end{proof}

Expressiveness requires the converse inequality of the one in Theorem~\ref{thm:adequate},
but only on one component: the carrier $\Phi$ of the initial algebra.
Further, the conditions include that $(\overline{B},B)$ is a fibration map. In particular,
for the canonical relation lifting $\Rel(B)$ this means that $B$ should preserve weak pullbacks;
this case is explained in more detail in Section~\ref{sec:canonical}.

\begin{lemma}\label{lm:iso-q}
	Let $\iota \colon X \stackrel{\cong}{\rightarrow} Y$ be an isomorphism in $\Cat{D}$. 
	Then $Q(\iota^{-1})^* (\lift{Q}Y) = \lift{Q}X$.
\end{lemma}
\begin{proof}
	Since $p \circ \lift{Q} = Q$, we have that $\lift{Q}(\iota^{-1}) \colon \lift{Q}X \rightarrow \lift{Q}Y$ is above $Q(\iota^{-1})$, 
	and hence $\lift{Q}X \leq Q(\iota^{-1})^*(\lift{Q}Y)$ by the latter's universal property. 
	For the converse, consider the following composition, where the left-hand side is the Cartesian morphism:
	$$
	\xymatrix@C=1.5cm{
		Q(\iota^{-1})^*(\lift{Q}Y) \ar[r]^-{\widetilde{Q\iota^{-1}}_{\overline{Q}Y}}
			& \lift{Q}Y \ar[r]^{\lift{Q} \iota}
			& \lift{Q}X
	}\,.
	$$
	This is above the identity on $QX$: $p(\lift{Q}(\iota) \circ \widetilde{Q\iota^{-1}}_{\overline{Q}Y}) = 
	p(\lift{Q}(\iota)) \circ p(\widetilde{Q\iota^{-1}}_{\overline{Q}Y}) = Q(\iota) \circ Q(\iota^{-1}) = \id_{QX}$. Hence we get 
	$Q(\iota^{-1})^* (\lift{Q}Y) \leq \lift{Q}X$ as needed. 
\end{proof}

\begin{theorem}\label{thm:expressive}
	Suppose $(\lift{B},B)$ is a fibration map. 
	If $\delta_{\Phi}^* (\lift{Q} L\Phi) \leq \lift{B}\, \lift{Q}\Phi$, then $(L,\delta)$ is expressive. 
\end{theorem}
\begin{proof}
	Let $(X,\gamma)$ be a $B$-coalgebra, with $\theory$ the associated theory map.
	We show that $\theory^*(\lift{Q}\Phi)$ is a post-fixed point of $\gamma^* \circ \lift{B}_X$:
	\begin{align*}
		\theory^*(\lift{Q}\Phi)
		&= (Q(\alpha^{-1}) \circ \delta_{\Phi} \circ B\theory \circ \gamma)^*(\lift{Q}\Phi) \\
		&= \gamma^* \circ (B\theory)^* \circ \delta_{\Phi}^* \circ Q(\alpha^{-1})^* (\lift{Q}\Phi) \\
		&= \gamma^* \circ (B\theory)^* \circ \delta_{\Phi}^* (\lift{Q} L\Phi) 
		\tag{Lemma~\ref{lm:iso-q}} \\
		&\leq \gamma^* \circ (B\theory)^* (\lift{B}\,\lift{Q}\Phi) \tag{assumption} \\
		&= \gamma^* \circ \lift{B}_X \circ \theory^*(\lift{Q} \Phi) \tag{$(\lift{B},B)$ fibration map} 
	\end{align*}
	Expressiveness follows since $\nu(\gamma^* \circ \lift{B}_X)$ is the greatest post-fixed point. 
\end{proof}

Note that in the above theorem the reference to the initial algebra $\Phi$ could be avoided
by requiring that the inequality in the assumption holds for arbitrary objects in $\Cat{D}$.
We opted for the above formulation reflecting the fact that,
whenever one is applying the theorem to concrete instances, it is useful 
that one is able to focus on the initial $L$-algebra only.

\subsection{Adequacy and Expressiveness w.r.t.\ Behavioural Equivalence}
\label{sec:canonical}

In the setting of coalgebraic modal logic recalled in Section~\ref{sec:logic}, 
Klin~\cite{Klin07} proved that 
\begin{enumerate}
	\item the theory map $\theory$ of a coalgebra $(X,\gamma)$ factors through coalgebra morphisms from $(X,\gamma)$;
	\item if $\delta$ has monic components, then $\theory$ factors as a coalgebra morphism followed by a mono.
\end{enumerate}
The first item can be seen as adequacy w.r.t.\ behavioural equivalence (i.e., identification by a coalgebra morphism), 
and the second as expressiveness.\footnote{For weak pullback preserving functors, behavioural equivalence coincides with bisimilarity but for arbitrary set functors
the latter can be a strictly smaller relation.}

In the current section we revisit this result for $\Set$ functors, as a sanity check of Definition~\ref{def:adequate-expressive}.
To obtain the appropriate notion of adequacy and expressiveness, we need to compare 
collections of formulas for equality. Therefore, the functor $\overline{Q}$ in Definition~\ref{def:adequate-expressive}
will be instantiated with $\overline{Q}X = (QX, \Delta_{QX})$ where $\Delta_{QX}$ denotes the diagonal. Then, for a coalgebra $(X,\gamma)$,
$\theory^*(\lift{Q}\Phi)$ is the set of all pairs of states $(x,y)$ such that $\theory(x) = \theory(y)$. Adequacy
then means that for every coalgebra $(X,\gamma)$, behavioural equivalence is contained in $\theory^*(\lift{Q}\Phi)$, i.e.,
if $x$ is behaviourally equivalent to $y$ then $\theory(x) = \theory(y)$. Expressiveness is the converse implication. 

We start with an abstract result, where the functor $\lift{Q}$ assigns the equality relation (diagonal); thus this is specifically
about capturing (behavioural) equivalence logically. 
To state and prove it, let $\Eq \colon \Set \rightarrow \Rel$
be the functor given by $\Eq(X) = \Delta_X$. 
This functor has a left adjoint $\Quot \colon \Rel \rightarrow \Set$, which maps
a relation $R \subseteq X \times X$ to the quotient of $X$ by the least equivalence relation containing $R$ (cf.~\cite{HJ98}). 
This can be generalised to the notion of fibration with quotients, see~\cite{Jacobs:fib}, but we stick to $\Set$ here.

\begin{proposition}\label{prop:equality}
	Consider the relation fibration $p \colon \Rel \rightarrow \Set$,
	let $B \colon \Set \rightarrow \Set$ be a functor with a lifting $\lift{B} \colon \Rel \rightarrow \Rel$ 
	which preserves diagonals, that is, $\lift{B} \circ \Eq= \Eq \circ B$, and such that 
	$(\lift{B},B)$ is a fibration map. 
	
	Let $P \dashv Q \colon \Set \leftrightarrows \Cat{D}^{\op}$ for some category $\Cat{D}$, 
	$L \colon \Cat{D} \rightarrow \Cat{D}$ a functor with an initial algebra
	and $\delta \colon BQ \Rightarrow QL$. 
	Then 
	\begin{enumerate}
		\item $(L,\delta)$ is adequate w.r.t.\ $\lift{B}$;
		\item if $\delta$ is componentwise injective, then $(L,\delta)$ is expressive w.r.t.\ $\lift{B}$,
	\end{enumerate}
	via $\overline{Q} = \Eq \circ Q$.
\end{proposition}
\begin{proof}
	For adequacy, we use Theorem~\ref{thm:adequate}. 
	By composition of adjoints, $P \circ \Quot$ is a left adjoint to $\Eq \circ Q$. 
	It will be useful to simplify $\lift{B} \circ \Eq \circ Q X$ and $\delta_X^* (\Eq \circ Q \circ L X)$:
	\begin{align}
		\lift{B} \circ \Eq \circ Q X & = \Delta_{BQX}  
		\,,  \\
		\delta_X^* (\Eq \circ Q \circ L X) &= (\delta_X \times \delta_X)^{-1}(\Delta_{QLX}) \,,
	\end{align}
	using that 
	$\lift{B}$ preserves diagonals in the first equality.
	The remaining hypothesis of Theorem~\ref{thm:adequate}
	is that $\lift{B} \circ \Eq \circ QX \leq \delta_X^* (\Eq \circ Q \circ L X)$ for all $X$, i.e.,
	$\Delta_{BQX} \subseteq (\delta_X \times \delta_X)^{-1}(\Delta_{QLX})$, which is trivial. 
	
	For expressiveness, we use Theorem~\ref{thm:expressive}. 
	By assumption, $(\lift{B},B)$ is a fibration map.
	We need to prove that $\delta_{\Phi}^* (\Eq \circ Q \circ L \Phi) \leq \lift{B} \circ \Eq \circ Q\Phi$, 
	which amounts to the inclusion
	$$
	  (\delta_\Phi \times \delta_\Phi)^{-1}(\Delta_{QL\Phi}) \subseteq \Delta_{BQ\Phi}
	$$
	But this is equivalent to injectivity of $\delta_\Phi$. 
\end{proof}
The canonical lifting $\Rel(B)$ of a $\Set$ functor $B$ always preserves diagonals, and
if $B$ preserves weak pullbacks, then it is a fibration map. Thus, we obtain
expressiveness w.r.t.\ bisimilarity for weak pullback preserving functors, if $\delta$
has injective components. 

In order to be able to cover a larger class of functors, and move to behavioural equivalence, 
we use the notion of \emph{lax extension} preserving diagonals. 

\begin{defiC}[\cite{mave15:laxe}]\label{def:lax}
     Let $B \colon \Set \to \Set$ be a functor. An operation $\llift{B}$ associating with each relation $R \subseteq X \times Y$
     a relation $\llift{B} \subseteq BX \times BY$ is called a {\em lax extension of $B$ preserving diagonals} if for all relations $R$, $S$,
     all functions $f$ an all sets $X$ we have 
     \begin{enumerate}
	  \item $\llift{B} (R^\bullet) = (\llift{B} R)^\bullet$ where $(\_)^\bullet$ denotes the converse relation,
	  \item $R \subseteq S$ implies $\llift{B}R \subseteq \llift{B}S$ (monotonicity),
	  \item $\llift{B}R;\llift{B}S \subseteq \llift{B}(R;S)$ (lax preservation of relational composition),
	  \item $\llift{B} \Gr(f) = \Gr (Bf)$ and, in particular, $\llift{B} \Delta_X = \Delta_{B X}$ (preservation of diagonals).
     \end{enumerate}
     Here for a function $f \colon X \to Y$  we denote by $\Gr(f)$ its graph relation:
     \[ \Gr(f) = \{ (x, f(x)) \mid x \in X \} \subseteq  X \times Y \] 
\end{defiC}
The following key fact is an immediate consequence of the results in~\cite{mave15:laxe}.
\begin{fact}
    Let $B$ be a set functor and let $\llift{B}$ be a lax extension of $B$ preserving diagonals. 
    Then on any coalgebra $(X,\gamma)$ we have that behavioural equivalence is 
    equal to $\nu (\gamma^* \circ \llift{B}_X)$.
\end{fact}
\begin{proof}[Proof (Sketch)]
    Monotonicity of $\llift{B}$ implies that $\gamma^* \circ \llift{B}_X$ is a monotone operator.
    The result now follows from $\Rel_X$ being a complete lattice and behavioural equivalence being the greatest
    post-fixed point of $\gamma^* \circ \llift{B}_X$. The latter is a consequence of~\cite[Prop.~9]{mave15:laxe}.
\end{proof}

In particular, for a weak pullback preserving functor $B$, the canonical lifting $\Rel(B)$ is a lax extension preserving diagonals. 
But the results in~\cite{mave15:laxe} also show that non-weak pullback preserving set functors have such lax extensions. In fact, 
any finitary functor for which an expressive logic with ``monotone'' modalities exist, has a suitable lifting. 
Examples include the so-called $(\_)^3_2$-functor , the functor $\Pow_n$ that maps a set $X$ to the collection $\Pow_n X$ of
subsets of $X$ with less than $n$ elements and the so-called monotone neighbourhood functor (cf.~Example~7 in \cite{mave15:laxe}).
The following proposition establishes that the lax lifting $\llift{B}$ fits into the fibrational framework of our paper,
and that Proposition~\ref{prop:equality} applies. 

\begin{proposition}
    Let $B \colon \Set \to \Set$ be a functor and let $\llift{B}$ be a lax lifting of $B$ that preserves diagonals.
    Then $\llift{B} \colon\Rel \to \Rel$ is a lifting of $B$ along the relation fibration $p \colon\Rel \to \Set$.
    In addition to that, $(\llift{B},B)$ is a fibration map.
\end{proposition}
\begin{proof}
    In order to turn $\llift{B}$ into a functor $\Rel \to \Rel$ we define
    $\llift{B}(f) \coloneqq Bf$ - we will verify later in the proof that
    the functor is well-defined. 
    Now note that for all relations $R \subseteq X \times X$ and functions $f\colon Y \to X$ we have
     \[ \llift{B}_Y(f^*(R))) = \llift{B}_Y((f \times f)^{-1}(R)) \stackrel{\mbox{\tiny (*)}}{=}  (Bf \times Bf)^{-1}(\llift{B}_X(R)) = (Bf)^*(\llift{B}_X(R)) \]
     where (*) is a well-known property of lax extensions (cf.~e.g. Remark~4 in~\cite{mave15:laxe}) and the other equalities follow
     from the definition of reindexing. 
     This implies that $(\llift{B},B)$ is a fibration map once we establish that $\llift{B}$ is a lifting of $B$
     along $p \colon \Rel \to \Set$.
   For the latter we only need to verify that $\llift{B}$ is a functor on $\Rel$.
   To avoid confusion, please note that~\cite{mave15:laxe} uses a different category $\Rel$ where
   the relations are morphisms whereas in our case the relations are objects.
   In order to see that $\llift{B}$ is well-defined on
   $\Rel$-morphisms, consider relations $R \subseteq X \times X$, $S\subseteq Y \times Y$ 
   and a function $f \colon R \to S \in \Rel$.
   We need to show that $\llift{B}(f) \colon \llift{B}(R) \to   \llift{B}(S)$. As  $\llift{B}(f) = Bf$, 
   we need to prove that $B f$ is a $\Rel$-morphism
   from $\llift{B}(R)$ to $\llift{B}(S)$. Consider an arbitrary pair $(t_1,t_2) \in \llift{B}(R)$.
   We have
   \[ \llift{B}(R)  \subseteq \llift{B}((f \times f)^{-1}(S)) = (Bf \times Bf)^{-1}(\llift{B}(S)) \]
   where the inclusion is a consequence of $f$ being a $\Rel$-morphism and monotonicity of $\llift{B}$, and
   the equality is an instance of (*). Therefore $(t_1,t_2) \in \llift{B}(R)$ implies
   $(t_1,t_2) \in  (Bf \times Bf)^{-1}[\llift{B}(S)]$ which is in turn equivalent to
   $(Bf(t_1),Bf(t_2))\in \llift{B}(S)$. This shows that $\llift{B}f \colon \llift{B}(R) \to \llift{B}(S)$ as required.
   Functoriality now follows easily from the fact that $\llift{B}f= Bf$ for all functions $f$.   
\end{proof}

\section{Examples}
\label{sec:examples}

In this section we instantiate the abstract framework to three concrete examples:
a behavioural metric on deterministic automata (Section~\ref{sec:sdw-distance}),
captured by $[0,1]$-valued tests; a unary predicate on transition systems (Section~\ref{sec:div});
and similarity of transition systems, captured by a logic with conjunction 
and diamond modalities (Section~\ref{sec:simulation}).

\subsection{Shortest distinguishing word distance}
\label{sec:sdw-distance}

We study a simple behavioural distance on deterministic automata: for two states $x,y$
and a fixed constant $c$ with $0 < c < 1$, the distance is given by $c^n$, where
$n$ is the length of the smallest word accepted from one state but not the other. 
Following~\cite{Bonchi0P18}, we refer to this distance as the \emph{shortest distinguishing word distance},
and, for an automaton with state space $X$, denote it by $\dsdw \colon X\times X \rightarrow [0,1]$. 

Formally, fix a finite alphabet $A$, and consider the functor $B \colon \Set \rightarrow \Set$,
$BX = 2 \times X^A$ of deterministic automata. We make use of the fibration $p \colon \Rel_{[0,1]} \rightarrow \Set$,
and define the lifting $\overline{B} \colon \Rel_{[0,1]} \rightarrow \Rel_{[0,1]}$ by
$$
\overline{B}(X,d) = \left( BX, ((o_1, t_1), (o_2, t_2)) \mapsto 
\begin{cases}
	1 & \text{ if } o_1 \neq o_2 \\
	c \cdot \max_{a \in A} \{ d(t_1(a), t_2(a))\} & \text{ otherwise}
\end{cases}
\right)
$$
The shortest distinguishing word distance $\dsdw$ on a deterministic automaton $\gamma \colon X \rightarrow 2 \times X^A$
is the greatest fixed point $\nu(\gamma^* \circ \lift{B}_X)$ (recall that in $\Rel_{[0,1]}$ we use the reverse order on $[0,1]$,
see Example~\ref{ex:relv-fib}). 

For an associated logic, we simply use words over $A$ as formulas, and define a satisfaction relation
which is weighted in $[0,1]$. Consider the following setting.
$$
\xymatrix@C=0.5cm{
\Set\ar@/^2ex/[rr]^-{P={[0,1]}^-} \save !L(.5) \ar@(dl,ul)^{B = 2 \times \Id^A} \restore 
& \bot 
& \Set^{\op} \ar@/^2ex/[ll]^-{Q=[0,1]^-} \save !R(.5) \ar@(ur,dr)^{L = A \times \Id + 1} \restore 
& & &\mbox{\qquad with} &
B([0,1]^-) \ar@{=>}[r]^-{\delta} & [0,1]^{L-}
}
$$
The initial algebra of $L$ is the set of words $A^*$. The natural transformation
$\delta$ is given by $\delta_X \colon 2 \times ([0,1]^X)^A \rightarrow [0,1]^{A \times X + 1}$, 
$$
\delta_X(o,t)(u) = 
\begin{cases}
	o & \text{ if } u= \ast \in 1 \\
	c \cdot t(a)(x) & \text{ if } u = (a,x) \in A \times X
\end{cases}
$$
which is a quantitative, discounted version of the Boolean-valued logic in Example~\ref{ex:logic-dfa}.
The logic $(L,\delta)$ defines, for any deterministic automaton $\langle o, t \rangle \colon X \rightarrow 2 \times X^A$, 
a theory map $\theory \colon X \rightarrow [0,1]^{A^*}$, given by 
\begin{align*}
	\theory(x)(\varepsilon) &= o(x) 
	\qquad \text{ and } \qquad \theory(x)(aw) = c \cdot \theory(t(x)(a))(w)\,,
\end{align*}
for all $x \in X$, $a \in A$, $w \in A^*$. 

We characterise the shortest distinguishing word distance with the above logic, by instantiating
and proving adequacy and expressiveness. Define
$$
\lift{Q} \colon \Set^\op \rightarrow \Rel_{[0,1]}\,, \qquad 
\lift{Q}(X) = \left( [0,1]^X, (\phi_1, \phi_2) \mapsto \sup_{x \in X} |\phi_1(x) - \phi_2(x)| \right) \,.
$$
Technically, this functor is given by mapping a set $X$ to the $X$-fold product of the 
object $\lift{[0,1]} = ([0,1], (r,s) \mapsto |r-s|)$. It follows immediately that $\lift{Q}$ has a left
adjoint, mapping $(X,d)$ to $\Hom((X,d),\lift{[0,1]})$, see Equation~\ref{eq:adjunction-simple}. 
This will be useful for proving adequacy below. 

The functor $\lift{Q}$ yields a `logical distance' between states $x,y \in X$, given by $\theory^*(\lift{Q} \Phi)$.
We abbreviate it by $\dlog \colon X \times X \rightarrow [0,1]$. Explicitly, we have  
\begin{equation}\label{eq:dlog-sdw}
\dlog(x,y) = \sup_{w \in A^*} |\theory(x)(w) - \theory(y)(w)| \,.
\end{equation}
Instantiating Definition~\ref{def:adequate-expressive}, the logic $(L,\delta)$ is 
\begin{itemize}
	\item \emph{adequate} if $\dsdw \geq \dlog$, and
	\item \emph{expressive} if $\dsdw \leq \dlog$. 
\end{itemize}
Here $\leq$ is the usual order on $[0,1]$, with $0$ the least element (the order in $\Rel_{[0,1]}$ is reversed).

To prove adequacy and expressiveness, we use Theorem~\ref{thm:adequate} and Theorem~\ref{thm:expressive}. 
The functor $\overline{Q}$ has a left adjoint,
as explained above. Further, $(\overline{B}, B)$ is a fibration map~\cite{Bonchi0P18}.
We prove the remaining hypotheses of both propositions by showing the equality 
$\overline{B} \, \overline{Q} X = \delta_X^*(\overline{Q} L X)$ for every object $X$ in $\Cat{D}$. 
To this end, we compute (suppressing the carrier set $BQX$):
$$
\begin{array}{rcl}
		& & \delta_X^*(\lift{Q} L X) \\
		&=& \left( ((o_1, t_1), (o_2, t_2)) \mapsto \sup_{u \in A\times X + 1} |\delta_X(o_1,t_1)(u) - \delta_X(o_2,t_2)(u)| \right) \\
		&=& \left( ((o_1, t_1), (o_2, t_2)) \mapsto 
		\begin{cases}
			1 & \text{ if } o_1\neq o_2 \\
			\sup_{u \in A\times X} |\delta_X(o_1,t_1)(u) - \delta_X(o_2,t_2)(u)|) & \text{ otherwise} 
		\end{cases}
		\right) \\
		&=& \left( ((o_1, t_1), (o_2, t_2)) \mapsto 
		\begin{cases}
			1 & \text{ if } o_1\neq o_2 \\
			\sup_{(a,x) \in A\times X} |c \cdot t_1(a)(x) - c \cdot t_2(a)(x)|) & \text{ otherwise} 
		\end{cases}
		\right) \\
		&=& \left( ((o_1, t_1), (o_2, t_2)) \mapsto 
		\begin{cases}
			1 & \text{ if } o_1\neq o_2 \\
			c \cdot \max_{a \in A}  \sup_{x \in X} |t_1(a)(x) - t_2(a)(x)|) & \text{ otherwise} 
		\end{cases}
		\right) \\
		&=& \lift{B} \, \lift{Q} X \\
\end{array}
$$
Hence, the logic $(L,\delta)$ is adequate and expressive w.r.t.\ the shortest distinguishing word 
distance, i.e., $\dsdw$ coincides with the logical distance $\dlog$ given in Equation~\ref{eq:dlog-sdw}.

\subsection{Divergence of processes}
\label{sec:div}

A state of an LTS is said to be \emph{diverging} if there exists an infinite path of $\tau$-transitions
starting at that state. To model this predicate, 
let $B \colon \Set \rightarrow \Set$, $BX = (\Powf X)^A$, where $A$ is a set
of labels containing the symbol $\tau \in A$. 
Consider the predicate fibration $p \colon \Pred \rightarrow \Set$,
and define the lifting $\lift{B} \colon \Pred \rightarrow \Pred$ by 
$$
\lift{B}(X,\Gamma) = ((\Powf X)^A, \{t \mid \exists x \in \Gamma. \, x \in t(\tau)\}) \,.
$$
The coinductive predicate defined by $\lift{B}$ on a $B$-coalgebra $(X,\gamma)$
is the set of diverging states:
$$
\nu(\gamma^* \circ \lift{B}_X) = (X, \{x \mid x \text{ is diverging}\})\,. 
$$

Now, we want to prove in our framework of adequacy and expressiveness that $x$ is diverging iff for every $n \in \mathbb{N}$ 
there is a finite path of $\tau$-steps starting in $x$, i.e., $x \models \diam{\tau}^n \top$ for every $n$. 
The proof relies on two main observations:
\begin{itemize}
	\item if $x$ satisfies infinitely many formulas of $\diam{\tau}^n \top$, then one of its $\tau$-successors does, too;
	\item if a state $x$ satisfies $\diam{\tau}^n \top$ for some $n$, then
	  $x$ satisfies $\diam{\tau}^m \top$ for all $0 \leq m \leq n$.
\end{itemize}
Combined, one can then give a coinductive proof, showing that if the current state
satisfies all formulas of the form $\diam{\tau}^n \top$, then one of its $\tau$-successors
also satisfies all these formulas. 

We make this argument precise by casting it into the abstract framework. First, for the logic, 
we have the following setting:
$$
\xymatrix@C=0.5cm{
\Set\ar@/^2ex/[rr]^-{P = 2^-} \save !L(.5) \ar@(dl,ul)^{B = (\Powf -)^A} \restore 
& \bot 
& \Pos^{\op} \ar@/^2ex/[ll]^-{Q = \Hom(-,2)} \save !R(.5) \ar@(ur,dr)^{L = \Id_{\top}} \restore 
& & \mbox{\qquad with} &
B \Hom(-,2) \ar@{=>}[r]^-{\delta} & \Hom(L-,2)
}
$$
Here $\Pos$ is the category of posets and monotone maps, and $2=\{0,1\}$ is the poset given 
by the order $0 \leq 1$. For a poset $S$,
$\Hom(S,2)$ is then the set of \emph{upwards closed} subsets of $S$.

The functor $LS = S_\top$ is defined on a poset $S$ by adjoining a new top element $\top$, i.e.,
the carrier is $S + \{\top\}$ and $\top$ is strictly above all elements of $S$. 
The initial algebra $\Phi$ of $L$ is the set of natural numbers, representing 
the formulas of the form $\diam{\tau}^n \top$, linearly ordered, with $0$ the top element. 
The choice of $\Pos$ means that the set 
$\Hom(\Phi,2)$ used to represent the theory of a state $x \in X$
consists of upwards closed sets (so closed under lower natural numbers in the usual ordering), corresponding to the second
observation above concerning the set of formulas satisfied by $x$.

The natural transformation $\delta$ is given by
$\delta_S \colon (\Powf \Hom(S,2))^A \rightarrow \Hom(S_\top, 2)$, 
$$
\delta_S(t)(x) = 
	\begin{cases} 
		1 & \text{ if } x = \top \\
		\bigvee_{\phi \in t(\tau)} \phi(x) & \text{ otherwise}
	\end{cases}
	\,.
$$
To show that this is well-defined, suppose $x,y \in S_\top$ with $x \leq y$,
and suppose $\delta_S(t)(x) = 1$. If $x=\top$, then $y=\top$, so $\delta_S(t)(y) =1$.
Otherwise, there is $\phi \in \Hom(S,2)$ such that $\phi \in t(\tau)$ and $\phi(x) = 1$. 
Since $\phi$ is upwards closed, $\phi(y)=1$ and consequently $\delta_S(t)(y)=1$ as needed. 

Now, the theory map $\theory \colon X \rightarrow \Hom(\Phi, 2)$ is given by
$\theory(x)(n) = 1$ iff there exists a path of $\tau$-steps of length $n$ from $x$.
We define 
$$
	\lift{Q} \colon \Pos^\op \rightarrow \Pred \,, \quad \lift{Q}(S) = (\Hom(S,2), \{\phi \mid \forall x \in S. \, \phi(x) = 1\}) \,.
$$
Instantiating Definition~\ref{def:adequate-expressive}, \emph{adequacy} means that if $x$ is diverging, 
then $x \models \diam{\tau}^n \top$ for all $n$; and expressiveness is the converse. 

We start with proving adequacy, using Theorem~\ref{thm:adequate}. The left adjoint 
$\overline{P}$ is given by
$\overline{P}(X,\Gamma) = (\Hom((X,\Gamma), (2, \{1\})), \{(\phi_1, \phi_2) \mid \forall x \in X. \, \phi_1(x) \leq \phi_2(x)\})$. 
It remains to prove that $\overline{B} \, \overline{Q}(S) \leq  \delta_S^*(\overline{Q} L S)$ for
all $S$. To this end, we observe $BQS = (\Powf(\Hom(S,2)))^A$ and compute:
\begin{align*}
\delta_S^*(\overline{Q} L S) 
	&= \{t \mid \delta_S(t) \in \overline{Q} L S\} \\
	&= \{t \mid \forall x \in S_\top. \, \delta_S(t)(x) = 1\} \\
	&= \{t \mid \forall x \in S. \, \delta_S(t)(x) = 1 \} \\
	&= \{t \mid \forall x \in S. \, \bigvee_{\phi \in t(\tau)} \phi(x) = 1\} 
\end{align*}
and $\overline{B} \, \overline{Q}(S) = \{t \mid (\lambda x. 1) \in t(\tau)\}$.
The needed inclusion is now trivial. 

For expressiveness we have to prove the reverse inclusion with $S = \Phi$, i.e., 
$$
\{t \in  (\Powf(\Hom(\Phi,2)))^A \mid \forall x \in \Phi. \bigvee_{\phi \in t(\tau)} \phi(x) = 1\}
\subseteq \{t \in (\Powf(\Hom(\Phi,2)))^A \mid (\lambda x. 1) \in t(\tau)\}.
$$
To this end, let $t$ be an element of the left-hand side, and suppose towards
a contradiction that for all $\phi$ with $\phi \in t(\tau)$, 
there is an element $x_\phi \in \Phi$ with $\phi(x_\phi) = 0$. 
Choosing an assignment $\phi \mapsto x_\phi$ of such elements, 
we get a \emph{finite} set $\{x_\phi \mid \phi \in t(\tau)\}$. 
Let $x_\phi$ be the smallest element of that set (w.r.t.\ the order of $\Phi$, i.e., the largest natural number),
and let $\psi \in t(\tau)$ 
be such that $\psi(x_\phi) = 1$; such a $\psi$ exists by assumption on $t$. 
However, since $x_\phi \leq x_\psi$ and $\psi$ is upwards closed
we have $\psi(x_\psi)=1$, which gives a contradiction. 
Hence, the inclusion holds as required. The lifting $(\overline{B},B)$ is a fibration map.
We thus conclude from~Theorem~\ref{thm:expressive} that the logic is expressive.

\subsection{Simulation of processes}
\label{sec:simulation}

Let $A$ be a set, and define the functor $B \colon \Set \rightarrow \Set$ by $BX = (\Powf X)^A$.
Let $\gamma \colon X \rightarrow (\Powf X)^A$ be $B$-coalgebra, i.e., a labelled transition system. 
Denote \emph{similarity} by ${\precsim} \subseteq X \times X$, defined more precisely below. 
Consider the logic with the following syntax: 
\begin{equation}\label{eq:logic-conj}
	\varphi,\psi \mathrel{::=} \diam{a}\varphi \mid \varphi \wedge \psi \mid \top
\end{equation}
where $a$ ranges over $A$, with the usual interpretation $x \models \varphi$ for states $x \in X$. 
A classical Hennessy-Milner theorem for similarity is:
\begin{equation}\label{eq:hm-sim}
x \precsim y \text{ iff } \forall \varphi. \, x \models \varphi \rightarrow y \models \varphi \,.
\end{equation}
We show how to formulate and prove this result within our abstract framework. 

First, recall from Equation~\ref{eq:bbar-simulation} in Section~\ref{sec:coind-pred} 
the appropriate lifting $\overline{B} \colon \Rel \rightarrow \Rel$ in the relation fibration 
$p \colon \Rel \rightarrow \Set$. A simulation on a $B$-coalgebra $(X,\gamma)$ is a relation $R$ such that $R \leq \gamma^* \circ \lift{B}_X(R)$,
and similarity $\precsim$ is the greatest fixed point of $\gamma^* \circ \lift{B}_X$. 

For the logic, to incorporate finite conjunction, we instantiate $\Cat{D}$
with the category $\SL$ of bounded (meet)-semilattices, i.e., sets equipped with an associative, commutative and idempotent 
binary operator $\wedge$ and a top element $\top$. 

To add the modalities $\diam{a}$ for each $a \in A$, we proceed as follows. 
Let $U \colon \SL  \rightarrow \Set$ be the forgetful functor. It has a left
adjoint $\mathcal{F} \colon  \Set \rightarrow \SL$, mapping a set
$X$ to the meet-semilattice $\Powf(X)$ with the top element given by $\emptyset$ and the meet by union. 
The functor $L \colon \SL \rightarrow \SL$ is given by $LX = \mathcal{F}(A \times UX)$; its initial algebra
$\Phi$ consists precisely of the language presented in Equation~\ref{eq:logic-conj}, quotiented by
the semilattice equations\footnote{To simplify the presentation we do not quotient with monotonicity axioms for the modal operators, i.e., we do not 
ensure that $\varphi_1 \leq  \varphi_2$ implies $\diam{a}\varphi_1 \leq \diam{a}\varphi_2$. 
}. For the adjunction, we use:
	$$
	\xymatrix@C=0.5cm{
	\Set\ar@/^2ex/[rr]^-{P = 2^-} \save !L(.5) \ar@(dl,ul)^{B = (\Powf -)^A} \restore 
	& \bot 
	& \SL^{\op} \ar@/^2ex/[ll]^-{Q=\Hom(-,2)} \save !R(.5) \ar@(ur,dr)^{L = \mathcal{F}(A \times U-)} \restore 
	& & \mbox{\qquad \qquad with} \quad
	B \Hom(-,2) \ar@{=>}[r]^-{\delta} & \Hom(L-,2)
	}
	$$
which is an instance of Equation~\ref{eq:adjunction-simple}.
Here $2=\{0,1\}$ is the meet-semilattice given 
by the order $0 \leq 1$. For a semilattice $S$,
the set $\Hom(S,2)$ of semilattice morphisms is isomorphic to the set of \emph{filters} on $S$: subsets $X \subseteq S$ such that 
$\top \in X$, and $x,y \in X$ iff $x \wedge y \in X$. 

To define the natural transformation $\delta_S \colon (\Powf(\Hom(S,2)))^A \rightarrow \Hom(\mathcal{F}(A \times US),2)$
on a semilattice $S$, we use that for every map $f \colon A \times US \rightarrow 2$ there is a unique semilattice
homomorphism $f^\sharp \colon \mathcal{F}(A \times US) \rightarrow 2$ extending it:
$$
\delta_S(t) = ((a,x) \mapsto \bigvee_{\phi \in t(a)} \phi(x))^\sharp 
	= \left( W \mapsto \bigwedge_{(a,x) \in W} \bigvee_{\phi \in t(a)} \phi(x) \right) \,.
$$
For an LTS $(X,\gamma)$, the associated theory map $\theory \colon X \rightarrow \Hom(\Phi, 2)$
maps a state to the formulas in~\eqref{eq:logic-conj} that it accepts, with the usual semantics. 

To recover~\eqref{eq:hm-sim}, we need to relate logical theories appropriately. 
Define 
$$
\overline{Q} \colon \SL^{\op} \rightarrow \Rel \,, \qquad
\overline{Q}S = (\Hom(S, 2), \{(\phi_1, \phi_2) \mid \forall x \in S. \, \phi_1(x) \leq \phi_2(x) \}) \,.
$$
Then $\theory^*(\overline{Q}\Phi) = \{(x,y) \mid \forall \varphi \in \Phi.\, \theory(x)(\varphi) \leq \theory(y)(\varphi)\}$, i.e., it relates all $(x,y)$ such that the set of formulas satisfied
at $x$ is included in the set of formulas satisfied at $y$. 
Thus, instantiating Definition~\ref{def:adequate-expressive}, 
adequacy ${\precsim} = \nu(\gamma^* \circ \overline{B}_X) \leq \theory^*(\overline{Q} \Phi)$
is the implication from left to right in Equation~\ref{eq:hm-sim}, and expressiveness is the converse. 

We prove adequacy and expressiveness. 
The functor $\overline{Q}$ has a left adjoint,
given by $\overline{P}(X, R) = \Hom((X, R), \overline{2})$, where $\overline{2} = (2, \{(x,y) \mid x \leq y\})$. 
This follows by Corollary~\ref{cor:polynomial} with Remark~\ref{rem:adjunction},
with $\SL$ as a full subcategory of the category of all algebras for the corresponding signature. 

Given a semilattice $S$, we compute $\delta_S^*(\lift{Q} L S) \subseteq (BQS)^2 = ((\Powf(\Hom(S,2)))^A)^2$:
\begin{align*}
	\delta_S^*(\lift{Q} L S) 
		&= \delta_S^*(\{(\phi_1,\phi_2) \mid \forall W \in \mathcal{F}(A \times US) . \, \phi_1(W) \leq \phi_2(W)\}) \\
		&= \{(t_1,t_2) \mid \forall W \in \mathcal{F}(A \times US) . \,
				\bigwedge_{(a,x) \in W} \bigvee_{\phi \in t_1(a)} \phi(x)
				\leq \bigwedge_{(a,x) \in W} \bigvee_{\phi \in t_2(a)} \phi(x)\} \,.
\end{align*}
Further, 
$
\lift{B}\,\lift{Q} S = \{(t_1,t_2) \mid \forall a \in A. \, \forall \phi_1 \in t_1(a). \, \exists \phi_2 \in t_2(a). \, 
							\forall x \in S. \, \phi_1(x) \leq \phi_2(x) \}
$.
For adequacy, we need to prove $\lift{B} \, \lift{Q} S \leq \delta_S^*(\lift{Q} L S)$; but this is trivial, 
given the above computations. For expressiveness, let $(t_1,t_2) \in \delta_S^*(\lift{Q} L S)$. 
We need to show that $(t_1,t_2) \in \lift{B} \, \lift{Q} S$. Suppose, towards a contradication,
that $(t_1,t_2) \not \in \lift{B} \, \lift{Q} S$, i.e., there exist $a \in A$ and $\phi_1 \in t_1(a)$ such that
for all $\phi_2 \in t_2(a)$, there is $x \in S$ with $\phi_1(x) = 1$ and $\phi_2(x) = 0$. We choose
such an element $x_{\phi_2}$ for every $\phi_2 \in t_2(a)$. Note that the collection 
$\{x_{\phi_2} \mid \phi_2 \in t_2(a)\}$ is \emph{finite}---here we make use of the image-finiteness
captured by the functor $B$. 
Now, consider the conjunction $\psi = \bigwedge_{\phi_2 \in t_2(a)} x_{\phi_2} \in S$. 
Using that $\phi_1$ is a homomorphism, we have 
$\phi_1(\psi) = \phi_1(\bigwedge_{\phi_2 \in t_2(a)} x_{\phi_2}) = \bigwedge_{\phi_2 \in t_2(a)} \phi_1(x_{\phi_2}) = 1$,
and consequently $\bigvee_{\phi \in t_1(a)} \phi(\psi) = 1$. 
We also have $\bigvee_{\phi \in t_2(a)} \phi(\psi) = \bigvee_{\phi \in t_2(a)} \bigwedge_{\phi_2 \in t_2(a)} \phi(x_{\phi_2}) = 0$
since $\phi_2(x_{\phi_2}) = 0$ for every $\phi_2 \in t_2(a)$. 
Finally, to arrive at a contradiction, let $W = \{(a, \psi)\}$. 
Since $(t_1,t_2) \in \delta_S^*(\lift{Q} L S)$ this implies 
$\bigvee_{\phi \in t_1(a)} \phi(\psi) \leq \bigvee_{\phi \in t_2(a)} \phi(\psi)$, which is
in contradiction with the above. 
It is easy to check that $(\overline{B},B)$ is a fibration map (cf.~\cite{HughesJ04}). Hence, we conclude
expressiveness from Theorem~\ref{thm:expressive}.

\begin{remark}
   In fact, the expressiveness argument also goes through if we replace $\SL$ in the above argument
   with the category of algebras for the bounded semilattice signature. As pointed out in Sec.~\ref{sec:adjunctions} this can be useful
   in cases where an axiomatisation of the class of algebras involved is not known. In the concrete case above
   we opted to work with the well-known category $\SL$ instead.
\end{remark}

\section{Finite-depth expressiveness and the Kleene fixed point theorem}
\label{sec:kleene}

In Section~\ref{sec:framework} we formulated expressiveness as an inequality 
$\nu(\gamma^* \circ \lift{B}_X) \geq \theory^*(\lift{Q} \Phi)$ for all $B$-coalgebras $(X,\gamma)$. The sufficient conditions 
formulated in Theorem~\ref{thm:expressive} ensure that 
$\theory^*(\lift{Q} \Phi)$ is a post-fixed point of $\gamma^* \circ \lift{B}_X$, so
that the desired inequality follows. Thereby, that approach relies on the Knaster-Tarski
fixed point theorem, constructing the greatest fixed point as the largest post-fixed point. 

In the current section we explore a different abstract technique for proving expressiveness,
which instead relies on a technique for constructing greatest fixed points which
is often referred to as Kleene's fixed point theorem. Given a monotone function
$\varphi \colon L \rightarrow L$ on a complete lattice $L$, we construct the chain
$$
\top \geq \varphi(\top) \geq \varphi(\varphi(\top)) \geq \ldots
$$
and take its limit
$$
\bigwedge_{i \in \mathbb{N}} \varphi^i(\top)
$$
If $\varphi$ preserves limits of $\omega$-cochains (decreasing sequences indexed by the natural numbers), also
referred to as cocontinuity, then this is the greatest fixed point of $\varphi$. 

This suggests a different route to expressiveness: we will formulate sufficient
conditions to ensure that 
$$
(\gamma^* \circ \lift{B}_X)^i(\top) \geq (\theory_i)^*(\lift{Q} \Phi_i)
$$
for all $i \in \mathbb{N}$; here $\Phi_i$ refers to formulas of modal depth at most $i$, made more precise
below using the \emph{initial} sequence of the functor $L$, and $\theory_i \colon X \rightarrow Q\Phi_i$ 
is the associated theory map.
The above family of inequalities (indexed by $i$) can be thought of as \emph{finite-depth expressiveness}: it 
states that the formulas of modal depth at most $i$ are expressive with respect to the $i$-th approximation
of the coinductive predicate defined by $\lift{B}$. For instance, that logical equivalence w.r.t.\ formulas
of depth at most $i$ in Hennessy-Milner logic imply $i$-step bisimilarity. 

These conditions are sufficient to ensure finite-depth expressiveness---if we then make the additional assumption
that $\gamma^* \circ \lift{B}_X$ is cocontinuous, we obtain proper expressiveness.
In the `Knaster-Tarski' approach to expressiveness of Theorem~\ref{thm:expressive}, instead, no such assumption is explicitly formulated. 
So in that approach, cocontinuity is not explicitly assumed.
This explains why in some of the examples---for instance similarity of labelled transition systems---part of the argument resembles a proof of cocontinuity. 

A remark is in order here. The cocontinuity of $\gamma^* \circ \lift{B}_X$,
which is a functor on a fibre (hence, a monotone map between posets), is of course different from preservation of limits of chains
by $\lift{B}$ or $B$. We refer to~\cite{HasuoKC18} for a proper study of the relation between these different sequences. 
The current section is primarily about another the relation between these various sequences and the initial sequence of $L$. 

Throughout this section we work again under Assumptions~\ref{as:basic} and~\ref{as:qbar}
concerning our overall categorical setting. 
We start by recalling the notion of initial and final sequence. 

\begin{definition}[Initial and final sequence]\label{def:sequence}
	Suppose $\Cat{D}$ has an initial object $0$. 
	The \emph{initial sequence} of $L$ is 
	the chain $(L^i0)_{i \in \mathbb{N}}$, with connecting morphisms $l_{i,j} \colon L^i 0 \rightarrow L^j 0$ for all $i \leq j$ defined by
	$l_{0,j} = {!_{L^j 0}}$ for all $j$ and $l_{i+1,j+1} = Ll_{i,j}$.
	Further, given an algebra $\alpha \colon LA \rightarrow A$, we inductively define a cocone
	$\alpha_i \colon L^i 0  \rightarrow A$ by $\alpha_0 = {!_A}$ and $\alpha_{i+1} = \alpha \circ L\alpha_i$.
	
	If $\Cat{C}$ has an initial object $1$, the final sequence of $B$ is defined dually
	as $(B^i 1)_{i \in \mathbb{N}}$, with the associated connecting morphisms $b_{j,i} \colon B^j 1 \rightarrow B^i 1$ for
	$i \leq j$. 
	Any coalgebra $\gamma \colon X \rightarrow BX$ defines a cone $\gamma_i \colon X \rightarrow B^i 1$
	by $\gamma_0 = {!_X}$ and $\gamma_{i+1} = B\gamma_i \circ \gamma$. 
	
\end{definition}
\newcommand{\colim}{\mathrm{colim}}
If $L$ preserves colimits of $\omega$-chains, of which the initial sequence is an instance, then the colimit $\colim_{i < \omega} L^i0$ carries an initial algebra~\cite{adamek1974free}. Dually if $B$ preserves limits of $\omega$-co-chains, $\lim_{i < \omega} B^i 1$ is a final coalgebra. 
In both cases, the elements of the respective sequences can be thought of as approximations of the initial algebra and final coalgebra, respectively.

For a coalgebra $(X,\gamma)$, we define the cone
$$
\Sem{\_}_i  = (P\gamma \circ \widehat{\delta}_X)_i  \colon L^i 0 \rightarrow PX \,,
$$
as in Definition~\ref{def:sequence} from the algebra $P\gamma \circ \widehat{\delta}_X \colon LPX \rightarrow PX$.
Let
$$\theory_i \colon X \rightarrow QL^i 0$$ 
be the transpose of $\Sem{\_}_i$. The elements of $L^i 0$ are thought of as modal formulas of rank at most $i$,
and $\theory_i$ is the theory map of a coalgebra restricted to those formulas.  
It is easy to show that the $\theory_i$ maps satisfy the following properties:

\begin{lemma}\label{lm:thmap-approx}
	For all $i$, the two triangles in the following diagram commute.
	$$
	\xymatrix{
		X \ar[rr]^{\theory_i} \ar[d]_-{\gamma} \ar[drr]^{\theory_{i+1}}
			& & QL^i 0 \\
		BX \ar[r]_-{B\theory_i}
			& BQL^i 0 \ar[r]_-{\delta_{L^i 0}}
			& QL^{i+1} 0 \ar[u]_-{Ql_{i,i+1}}
	}
	$$
\end{lemma}
Furthermore we define a sequence $(\delta_i \colon B^i 1 \rightarrow QL^i 0)_{i \in \mathbb{N}}$,
which iterates $\delta$ on the final sequence of $B$, 
as follows:
$$
\delta_0 = \id \colon 1 \rightarrow Q0 \qquad \delta_{i+1} = 
\left(\xymatrix{
	BB^i 1 \ar[r]^-{B\delta_i}
		& BQL^i 0 \ar[r]^-{\delta_{L^i 0}}
		& QL^{i+1}0 
}\right).
$$
We use here that $Q0 = 1$ is a final object, as $0$ is initial and $Q$ a right adjoint. 
This enables us to relate $\theory_i$ and $\gamma_i$.

\begin{lemma}\label{lm:theory-delta-gamma}
	Let $(X,\gamma)$ be a coalgebra. For all $i$, $\theory_i = \delta_i \circ \gamma_i$.
\end{lemma}
\begin{proof}
	By induction on $i$. The base case is trivial: $\theory_0 = {!_X} = \delta_0 \circ \gamma_0$. 
	Suppose it holds for some $i$. Then:
	$
		\theory_{i+1}
		= \delta_{L^i 0} \circ B\theory_i \circ \gamma 
		= \delta_{L^i 0} \circ B\delta_i \circ B\gamma_i \circ \gamma 
		= \delta_{i+1} \circ \gamma_{i+1} 
	$,
	where the first equality holds by Lemma~\ref{lm:thmap-approx}.
\end{proof}

The following lemma shows that,
for a coalgebra $(X,\gamma)$, the elements 
of the final sequence of $\gamma^*\circ \lift{B}_X$ (in the fibre $\Cat{E}_X$)
can be retrieved from the final sequence of $\lift{B}$ (in the total category $\Cat{E}$)
by reindexing along the maps $\gamma_i \colon X \rightarrow B^i 1$.

\begin{lemma}\label{lm:relate-final-seq}
	Suppose $(B,\lift{B})$ is a fibration map, and let $(X,\gamma)$ be a $B$-coalgebra.
	Then for all $i$, $(\gamma_{i})^*(\lift{B}^{i} 1) = (\gamma^* \circ \lift{B}_X)^i(\top)$. 
\end{lemma}
\begin{proof}
	By induction on $i$. 
	The base case is easy, since reindexing in $\CLat$-fibrations preserves top elements. 
	For the inductive case, suppose it holds for some $i$. We compute:
	\begin{align*}
		(\gamma_{i+1})^*(\lift{B}^{i+1} 1)
		&= (B\gamma_i \circ \gamma)^* (\lift{B} \, \lift{B}^i 1) \\
		&= \gamma^* \circ (B\gamma_i)^* (\lift{B} \, \lift{B}^i 1) \\
		&= \gamma^* \circ \lift{B}_X \circ (\gamma_i)^* (\lift{B}^i 1) \tag{$(\lift{B},B)$ fibration map} \\
		&= \gamma^* \circ \lift{B}_X \circ (\gamma \circ \lift{B}_X)^i(\top) \tag{induction hypothesis} \\
		&= (\gamma \circ \lift{B}_X)^{i+1}(\top) \,.
\qedhere
	\end{align*}
\end{proof}
The following result now establishes a sufficient condition on $\delta$ for finite-depth expressiveness,
formulated in terms of the final sequence of $\overline{B}$ and initial sequence of $L$.

\begin{proposition}[Finite-depth expressiveness]\label{prop:approx-expr}
	Suppose $(\overline{B}, B)$ is a fibration map. Then for all $i$: 
	if $\delta_i^*(\lift{Q} L^i 0) \leq \lift{B}^i 1$
	then for any $B$-coalgebra $(X,\gamma)$, we have $(\theory_i)^*(\lift{Q} L^i 0) \leq (\gamma^* \circ \lift{B}_X)^i(\top)$. 
\end{proposition}
\begin{proof}
	If $\delta_i^*(\lift{Q} L^i 0) \leq \lift{B}^i 1$, then:
	\begin{align*}
		(\theory_i)^*(\lift{Q} L^i 0) 
			&= (\gamma_i)^* \circ (\delta_i)^* (\lift{Q} L^i 0) \tag{Lemma~\ref{lm:theory-delta-gamma}}\\
			&\leq (\gamma_i)^* (\lift{B}^i 1) \tag{assumption}\\
			&= (\gamma \circ \lift{B}_X)^i(\top) \tag{Lemma~\ref{lm:relate-final-seq}} \,.
	\end{align*}
\end{proof}

A natural way to move from the above result on finite-depth expressiveness to full expressiveness
is to assume that the functors $\lift{B}_X$ on the fibre preserve 
limits of $\omega$-chains. Note that these are functors on fibres (that is, monotone functions),
and there is no assumption on $B$ or $\lift{B}$ preserving anything. 

\begin{theorem}[Expressiveness via Kleene fixed point theorem]\label{thm:expr-kleene}
	Suppose $(\overline{B}, B)$ is a fibration map, and for all $i$: 
	 $\delta_i^*(\lift{Q} L^i 0) \leq \lift{B}^i \lift{Q} 0$.
	 Then for any $B$-coalgebra $(X,\gamma)$ and for all $i$: $\theory^*(\lift{Q}\Phi) \leq (\gamma^* \circ \lift{B}_X)^i(\top)$. 
	 In particular, if $\lift{B}_X$ preserves limits of $\omega$-cochains,
	 then the logic $(L, \delta)$ is expressive. 
\end{theorem}
\begin{proof}
	Let $(X,\gamma)$ be a coalgebra; we need to prove that 
	$\theory^*(\lift{Q} \Phi) \leq \nu(\gamma^* \circ \lift{B}_X)$. 
	For the initial algebra $(\Phi, \alpha)$, we have the sequence $\alpha_i \colon L^i 0 \rightarrow \Phi$
	(Definition~\ref{def:sequence})
	and since  $\Sem{\_} \colon \Phi \rightarrow PX$ is an algebra morphism,
	it follows that $\Sem{\_} \circ \alpha_i = \Sem{\_}_i$ for all $i$. Hence, 
	$\theory_i = Q\alpha_i \circ \theory \colon X \rightarrow QL^i0$. 
	Now, since there is the composite morphism 
	$$
	\xymatrix@C=1.3cm{
		\theory^*(\lift{Q} \Phi) \ar[r]^-{\widetilde{\theory}_{\lift{Q}\Phi}}
			& \lift{Q} \Phi \ar[r]^-{\lift{Q}\alpha_i}
			& \lift{Q}	L^i 0
	}\,,
	$$
	above $Q\alpha_i \circ \theory = \theory_i$, we have 
	$\theory^*(\lift{Q} \Phi) \leq \theory_i^*(\lift{Q} L^i 0)$. Combined
	with Proposition~\ref{prop:approx-expr}, we get 
	$$
		\theory^*(\lift{Q} \Phi) \leq (\gamma^* \circ \lift{B}_X)^i(\top)
	$$
	for all $i$. 
	
	Finally, if $\lift{B}_X$ preserves limits of $\omega$-cochains, then 
	so does $\gamma^* \circ \lift{B}_X$ (reindexing $\gamma^*$ preserves all
	meets, by the definition of $\CLat$-fibration). Hence 
	$\bigwedge_{i \in \mathbb{N}} (\gamma^* \circ \lift{B}_X)^i(\top) = \nu(\gamma^* \circ \lift{B}_X)$.
	And thus $\theory^*(\lift{Q} \Phi) \leq  \nu(\gamma^* \circ \lift{B}_X)$.
\end{proof}

\begin{example}
	We show finite-depth expressiveness (Prop.~\ref{prop:approx-expr}) via the above approach for the example of similarity of labelled transition systems (Section~\ref{sec:simulation}).
	The relevant endofunctors $B = \Powf(-)^A,\lift{B}, L$, adjunctions $P \dashv Q$ and $\lift{P} \dashv \lift{Q}$, and $\delta$ 
	are all as defined there. Contrary to the 
	treatment in Section~\ref{sec:simulation}, with the current approach it matters
	quite a bit whether $A$ is finite or not.\footnote{This was pointed out to us by
	Yuichi Komorida.} For the moment, we will assume
	that $A$ is finite, which significantly simplifies the matter. 
	In Example~\ref{ex:now-infinite} below we discuss the infinite case. 
	
	The final sequence of $\lift{B}$ is concretely described as a sequence of relations on $B^i 1$, by just instantiating
	$\lift{B}$ for the inductive case:
	$$
	\lift{B} \, \lift{B}^i 1 = 
	\{(t_1,t_2) \mid \forall a \in A. \, \forall x \in t_1(a). \, \exists y \in t_2(a). (x,y) \in \lift{B}^i 1 \}\,. 
	$$
	and $\lift{B}^0 1 = \{(\ast, \ast)\} \subseteq 1 \times 1$. Thus, $(t_1, t_2) \in \lift{B}^i 1$ iff $t_1$ is ``$i$-step simulated'' by $t_2$, where
	both $t_1$ and $t_2$ are viewed as trees of height at most $i$. 
	
	The initial sequence of $L \colon \SL \rightarrow \SL$ is characterised, once again by spelling out the definition, by
	$L^0 0 = 0 = \{\top\}$ (the one-element semilattice, which is the initial object in $\SL$) and $LL^i 0 = \Powf(A \times L^i 0)$ (the
	free semilattice, see Section~\ref{sec:simulation}). Concretely, elements of $L^i 0$
	can be identified with formulas of depth at most $i$ in the logic of Section~\ref{sec:simulation} (diamond modalities and conjunction),
	quotiented by the semilattice equations. By the assumption that $A$ is finite, each set $L^i 0$ is finite. 
	
	We continue to prove the main hypothesis of Proposition~\ref{prop:approx-expr} and Theorem~\ref{thm:expr-kleene}: that
	\begin{equation}\label{eq:sim-kleene-ind}
		\delta_i^*(\lift{Q}L^i0) \leq \lift{B}^i 1
	\end{equation}
	for all $i$. Before doing so, we spell out $\delta_i^*(\lift{Q}L^i0)$ in some more detail. First, we 
	characterise $\delta_{i+1} \colon BB^i 1 \rightarrow QLL^i0$:
	\begin{align*}
		\delta_{i+1}(t)(W) 
		&= \delta_{L^i 0}(B\delta_i(t))(W) \\
		&= \bigwedge_{(a,\psi) \in W} \bigvee_{\phi \in (B\delta_i(t))(a)} \phi(\psi) \\
		&= \bigwedge_{(a,\psi) \in W} \bigvee_{x \in t(a)} \delta_i(x)(\psi) \\
	\end{align*}
	The map $\delta_i$ assigns to an element $t \in B^i 1$ the formulas of modal depth at most $i$ that hold for $t$,
	viewed as a tree.
	
	We now prove~\eqref{eq:sim-kleene-ind} by induction on $i$. The base case is trivial. For the inductive case,
	assume~\eqref{eq:sim-kleene-ind} holds for some $i$. 
	We have to prove that, for all $(t_1,t_2) \in \delta^*_{i+1}(\lift{Q}L^{i+1}0)$, $a \in A$ and $x \in t_1(a)$, there exists 
	$y \in t_2(a)$ such that $(x,y) \in B^i 1 $. 
	
	Spelling out the definition of $\delta_{i}^*(\lift{Q} L^{i}0)$ yields
	$$
		\delta_{i}^*(\lift{Q} L^{i}0)
		= \{(t_1, t_2) \in \lift{B}^{i}1 \mid \forall \phi \in L^{i} 0. \, \delta_{i}(t_1)(\phi) \leq \delta_{i}(t_2)(\phi)\} \,.
	$$
	For the $i+1$ case we can expand this further using the above equation for $\delta_{i+1}$:
	\begin{align*}
		\delta_{i+1}^*(\lift{Q} L^{i+1}0)
		&= \{(t_1, t_2) \mid \forall W \in L L^{i}0. \, \bigwedge_{(a,\psi) \in W} \bigvee_{x \in t_1(a)} \delta_i(x)(\psi)
		\leq \bigwedge_{(a,\psi) \in W} \bigvee_{x \in t_2(a)} \delta_i(x)(\psi) \}\,.
	\end{align*}

	Now, let $(t_1, t_2) \in \delta^*_{i+1}(\lift{Q}L^{i+1}0)$, $a \in A$ and $x \in t_1(a)$.  
	Let $\Psi = \{\psi \in L^i0 \mid \delta_i(x)(\psi) = 1\}$. Note that this is indeed a finite set, as $L^i 0$ is, so we will be able to consider its conjunction.
	Since $\delta_i(x)$ is a filter, we get $\delta_i(x)(\bigwedge_{\psi \in \Psi} \psi) = \bigwedge_{\psi \in \Psi} \delta_i(x)(\psi) = 1$, where the second equality holds by definition of $\Psi$.  
	Since
	 $(t_1, t_2) \in \delta^*_{i+1}(\lift{Q}L^{i+1}0)$, 
	 taking $W = \{(a, \bigwedge \Psi)\}$ (this represents $\diam{a}\bigwedge_{\delta_i(x)(\psi) = 1} \psi$)
	 we get that there exists $y \in t_2(a)$ such that $\delta_i(y)(\bigwedge_{\psi \in \Psi} \psi) = 1$. Hence, since $\delta_i(y)$ is a filter,
	 we get $\delta_i(y)(\psi) =1$ for all $\psi \in \Psi$. Thus $(x,y) \in \delta_{i}^*(\lift{Q}L^i 0)$,
	 and by the induction hypothesis we obtain $(x,y) \in \lift{B}^i 1$ as needed. 
\end{example}



The above proof relies on the assumption that the set of labels $A$ is finite. In the following example
we show a way to adapt the proof to the case where this assumption is dropped.

\begin{example}\label{ex:now-infinite}
	If $A$ is not assumed to be infinite, the above proof does not work, as
	the meet $\bigwedge \Psi$ may be infinite, which is not defined as we are
	working with semilattices. To remedy this, with $\Psi$ defined as above,
	let $\Psi_0 \subseteq \Psi_1 \subseteq \Psi_2 \subseteq \ldots$ be an increasing
	sequence of finite subsets of $\Psi$ such that $\bigcup_{i \in \mathbb{N}} \Psi_i = \Psi$.
	First note that, for each $i$, we have $\delta_i(x)(\bigwedge_{\psi \in \Psi_i} \psi) = 1$
	using again the filter property. Now, consider
	$$
	W_i = \{(a, \bigwedge \Psi_i) \}\,.
	$$
	for $i \in \mathbb{N}$.
	Following the earlier reasoning, we get for each $i$ an element $y \in t_2(a)$ such that 
	$\delta_i(y)(\bigwedge_{\psi \in \Psi_i} \psi) = 1$. In fact, since $t_2(a)$ is finite, there exists an $y$ 
	such that $\delta_i(y)(\bigwedge_{\psi \in \Psi_i} \psi) = 1$ for infinitely many $i$ (and as a consequence for all $i$). 
	Since $\delta_i(y)$ is a filter this means $\bigwedge_{\psi \in \Psi_i} \delta_i(y)(\psi) = 1$ for infinitely many
	$i$. Finally, since every $\psi \in \Psi$ is contained in some of these sets $\Psi_i$ we obtain the desired result
	that $\delta(y)(\psi) = 1$ for all $\psi \in \Psi$. The proof then concludes
	as above. 
\end{example}

\section{Future work}
\label{sec:conclusions}

We proposed suitable notions of expressiveness and adequacy,
connecting coinductive predicates in a fibration to coalgebraic modal logic in a contravariant adjunction.
Further, we gave sufficient conditions on the one-step semantics that guarantee expressiveness
and adequacy, and showed how to put these methods to work in concrete examples. 

There are several avenues for future work. First, an intriguing question is whether
the characterisation of behavioural metrics in~\cite{KonigM18,WildSP018}
can be covered in the setting of this paper, as well as logics for other distances
such as the (abstract, coalgebraic) Wasserstein distance. Those behavioural metrics
are already
framed in a fibrational setting~\cite{Bonchi0P18,SprungerKDH18,BaldanBKK18,KKHKH19}. 
While all our examples
are for coalgebras in $\Set$, the fibrational framework allows different base categories,
which might be useful to treat, e.g., behavioural metrics for continuous probabilistic systems~\cite{BreugelW05}.

A further natural question is whether we can automatically \emph{derive} logics for 
a given predicate. As mentioned in the introduction, there are various tools 
to find expressive logics for behavioural equivalence. But extending this to the current
general setting is non-trivial. Conversely, given a logic, one would like to associate
a lifting to it, perhaps based on techniques related to $\Lambda$-bisimulations~\cite{GorinS13,BakhtiariH17,Enqvist13}.

\bibliographystyle{alpha}
\bibliography{refs}

\end{document}